\documentclass[sigconf]{acmart}

\usepackage [ n, advantage , operators , sets , adversary , landau , probability , notions , logic , ff, mm, primitives , events , complexity , asymptotics , keys]{ cryptocode}
      
\usepackage{amsfonts}
\usepackage{color, float}
\usepackage{appendix}
\usepackage{amsmath,amsthm,bm}
\usepackage{amssymb, longtable}
\usepackage{graphicx,xypic,tikz,pgfplots}
\usepackage{algorithm}
\usepackage{algorithmicx}
\usepackage{algpseudocode}
\usepackage{url,hyperref,enumerate}
\usepackage{float}

\usepackage{hyperref}
\usepackage{xspace}

\newcommand{\F}{{\mathbb F}}
\newcommand{\Z}{{\mathbb Z}}
\newcommand{\R}{{\mathbb R}}
\newcommand{\Grp}{{\mathbb G}}

\newcommand{\A}{{\mathcal A }}
\newcommand{\Cp}{{\mathcal C }}
\newcommand{\DS}{{\mathcal D }}
\newcommand{\E}{\textrm{E}}
\newcommand{\D}{\textrm{D}}
\newcommand{\Enc}{\textsf{Enc}}
\newcommand{\Dec}{\textsf{Dec}}

\newcommand{\siml}{ \mathcal{S} }

\newcommand{\rbu}{\textbf{r}}

\newcommand{\xbu}{\textbf{x}}
\newcommand{\ybu}{\textbf{y}}
\newcommand{\zbu}{\textbf{z}}
\newcommand{\tbu}{\textbf{t}}
\newcommand{\Xbu}{\textbf{X}}

\newcommand{\gdbu}{{\bf err}}

\newcommand{\sbu}{{\bf s}}
\newcommand{\gd}{{\tt GD}}
\newcommand{\us}{\mathcal{U}}
\newcommand{\vs}{\mathcal{V}}
\newcommand{\func}{{\mathcal F}}

\newcommand{\thetabu}{{\bm \theta}}
\newcommand{\omegabu}{{\bm \omega}}
\newcommand{\mubu}{{\bm \mu}}
\newcommand{\sigmabu}{{\bm \sigma}}

 \newcommand{\regsys}{\textsf{PrivFL}}

\newtheorem{definition}{Definition}
\newtheorem{theorem}{Theorem}
\newtheorem{lemma}{Lemma}

\usepackage{pgfplots,subfig}
\usepgfplotslibrary{groupplots}

\makeatletter
\newcommand{\figcaption}{\def\@captype{figure}\caption}
\newcommand{\tabcaption}{\def\@captype{table}\caption}
\makeatother

\def\hlinewd#1{%
\noalign{\ifnum0=`}\fi\hrule \@height #1 %
\futurelet\reserved@a\@xhline}

\makeatletter
\def\hlinewd#1{%
  \noalign{\ifnum0=`}\fi\hrule \@height #1 \futurelet
   \reserved@a\@xhline}
\makeatother

\makeatletter
\def\BState{\State\hskip-\ALG@thistlm}
\makeatother

\makeatletter
\newcommand*{\bigcdot}{}
\DeclareRobustCommand*{\bigcdot}{%
  \mathbin{\mathpalette\bigcdot@{}}%
}
\newcommand*{\bigcdot@scalefactor}{.5}
\newcommand*{\bigcdot@widthfactor}{1.15}
\newcommand*{\bigcdot@}[2]{%
  \sbox0{$#1\vcenter{}$}
  \sbox2{$#1\cdot\m@th$}%
  \hbox to \bigcdot@widthfactor\wd2{%
    \hfil
    \raise\ht0\hbox{%
      \scalebox{\bigcdot@scalefactor}{%
        \lower\ht0\hbox{$#1\bullet\m@th$}%
      }%
    }%
    \hfil
  }%
}
\makeatother

\newcommand{\sstm}{\textbf{Share}_m^t}
\newcommand{\rectm}{\textbf{Rec}_m^t}

\usepackage{makecell}
\usepackage[linewidth=0.5pt]{mdframed}

\usepackage{diagbox}
\usepackage{xcolor}

\makeatletter
\def\@copyrightspace{\relax}
\makeatother

\setcopyright{none} 
\renewcommand\footnotetextcopyrightpermission[1]{}

\begin{document}
\fancyhead[LE,RO]{}

\title{\texttt{PrivFL}: Practical Privacy-preserving Federated Regressions on High-dimensional Data over Mobile Networks}
%
\author{Kalikinkar Mandal}
\affiliation{%
 \institution{University of Waterloo}
 \streetaddress{200 University Ave W.}
 \city{Waterloo}
 \state{Ontario}
 \postcode{N2L 3G1}}
\email{kmandal@uwaterloo.ca}
\author{Guang Gong}
\affiliation{%
 \institution{University of Waterloo}
 \streetaddress{200 University Ave W.}
 \city{Waterloo}
 \state{Ontario}
 \postcode{N2L 3G1}}
 \email{ggong@uwaterloo.ca}

\begin{abstract}
Federated Learning (FL) enables a large number of users to jointly learn a shared machine learning (ML) model, coordinated by a centralized server, where the data is distributed across multiple devices. 
This approach enables the server or users to train and learn an ML model using gradient descent, while keeping all the training data on users' devices. 
We consider training an ML model over a mobile network where {\it user dropout} is a common phenomenon.
Although federated learning was aimed at reducing data privacy risks, the ML model privacy has not received much attention. 

In this work, we present \regsys, a privacy-preserving system for training (predictive) linear and logistic regression models and oblivious predictions in the federated setting, while guaranteeing data and model privacy as well as ensuring robustness to users dropping out in the network. 
We design two privacy-preserving protocols for training linear and logistic regression models based on an additive homomorphic encryption (HE) scheme and an aggregation protocol. 
Exploiting the training algorithm of federated learning, at the core of our training protocols is a secure multiparty global gradient computation on alive users' data. 
We analyze the security of our training protocols against semi-honest adversaries. As long as the aggregation protocol is secure under the aggregation privacy game and the additive HE scheme is semantically secure,  \regsys\ guarantees the users' data privacy against the server, and the server's regression model privacy against the users. 
We demonstrate the performance of \regsys\ on real-world datasets and show its applicability in the federated learning system. 
\end{abstract}

\vspace{-0.15cm}

\maketitle

\section{Introduction}
Due to their powerful capabilities, machine learning (ML) algorithms are deployed in various applications, from mobile applications to massive-scale data centers for serving various tasks such as predictive analysis, classification, and clustering. Companies such as browser, telecom, web services, edge computing, and advertisement collect a huge amount of data from users to learn ML models for improving quality of services  and user experiences, and for providing intelligent services. 
%
%
Predictive analytics is a fundamental task in machine learning, which has many applications ranging from advertisement analytics to financial modeling, supply chain analysis, to health analytics. 
%
%
A common approach to the data analytics is that users send their data to a centralized server where the server executes machine learning algorithms on collected data. 
A major drawback of this approach is that the transparency of data processing at the server is not clear, and such centralized servers (or Artificial Intelligence (AI) platforms) are easy targets for hackers (e.g., AI.type \cite{AItype}, ). 
Recently the European Union's General Data Protection Regulation (GDPR) mandates companies to exhibit transparency in handling personal data, data processing, etc. 

Federated learning \cite{TowardsFL,McMahan-model-averaging} is a promising distributed machine learning approach, and its goal is to collaboratively learn a shared model, coordinated by a (centralized) server, while the training data remains on user devices. In federated learning, an ML model on users' data is learnt in an iterative way, which has four steps: 1) a set of users is chosen by the server to compute an updated model; 2) each user computes an updated model on its local data; 3) the updated local models are sent to the server; and 
4) the server aggregates these local models to construct a global model. With the advancement of 5G, the federated learning approach will be an attractive solution for machine leaning in edge computing.  

With the growing concern of data privacy, federated learning aimed at reducing data privacy risks by avoiding storing data at a centralized server. 
In privacy-preserving machine learning, two major privacy concerns are the user input privacy and the ML model privacy \cite{ModelStealingAttack,ModelInversionAttacks}. 
However,  the model privacy  in the federated learning setting has not received much attention. 
Besides the model leaking sensitive information about training data, the model privacy must also be protected against unauthorized use or misuse. For instance, a {\it covert user} belonging to one organization participating in a federated training process may disclose the trained model to another organization, which may lead to obtaining a better model by further training it on their dataset or may use it as a prediction service to do business. 
Although the training data remains on the device, we have found certain scenarios in which the local gradient computation {\it leaks private information} about locally stored data (see Section~\ref{sec:PrivLeakage}). 

In this work, we consider a federated machine learning setting in which a server that coordinates the machine learning process wishes to learn an ML model on a joint dataset belonging to a set of mobile users.
We consider both the mobile users' data privacy against the server, and the model privacy against the users in the federated setting. 
A {\it user dropout} is an important consideration in mobile applications, recently considered in \cite{SecureAggrePPML,secure-agg-cacr2018}. Especially, during the training phase which is a computationally intensive and time consuming task, it can happen any time from the network. 
Our goal is to design secure and dropout-robust training protocols for predictive analytics for mobile applications where, in this work, we consider two fundamental regression models, namely linear/ridge regression and logistic regression. 

In a recent work, the authors of \cite{SecureAggrePPML} considered a general problem of secure aggregation for privacy-preserving machine learning in the federated setting. 
In a follow-up work, the authors of \cite{secure-agg-cacr2018} presented an improved protocol for secure aggregation. 
No complete training protocol for any machine learning algorithm is developed in \cite{SecureAggrePPML,secure-agg-cacr2018}. 
Privacy-preserving training for linear and logistic regressions is not a new problem. 
Secure training protocols for linear regression have been proposed, e.g., in \cite{gasconpets2017,Nikolaenko:2013SP} where, unlike ours, a set of users participating in the training process is connected on a stable network. 
Secure training protocols for Machine-Learning-as-a-Service (MLaaS) for logistic regression are proposed in \cite{LogisReg,LogisReg-HE,LogisReg-Cloud,LogisticReg-HE-MedInform}. 
Although several solutions have been proposed for linear and logistic regressions using somewhat or fully homomorphic encryption (SWHE or FHE) \cite{fhe}, garbled circuit \cite{GarbledCircuit}, or hybrid techniques e.g., in \cite{ABY3}, until now, 
the regression model training over mobile networks has not been considered, under the scenario of users dropping out. 
Moreover, the suitability of expensive cryptographic primitives such as SWHE or FHE on mobile devices has not been studied. 
Our goal is to design protocols using lightweight crypto-primitives and schemes suitable for mobile devices in both training and prediction phases. 
\vspace{-0.2cm}
\subsection{Our Contributions}
We design, analyze, and evaluate \regsys, a system for privacy-preserving training and oblivious  prediction of regression models, namely linear/ridge regression and logistic regression in the federated setting.  
 
\par\noindent{\bf Dropout-robust regression training protocols.} 
We design two privacy-preserving protocols -- one is for a linear regression and another is for a logistic regression  -- for training a regression model over a mobile network while providing robustness in the event of users dropping out. 
In a nutshell, our privacy-preserving protocol for multiparty regression training consists of multiple (parallel) two-party shared local gradient computation protocols, followed by a global gradient share-reconstruction protocol  (see Section~\ref{sec:ProtocolFlow}). 
 In our protocol, the users and the server execute the following three steps: 
 1) a shared local gradient computation protocol is run between the server and a user to securely compute two (additive) shares of the local gradient on the user's data to prevent the input leakage, even if the user has a single data point;
%
  2) the server and all alive users execute an aggregation protocol to construct one share of the global gradient; and 3) the server computes the second share of the global gradient from its local gradient shares.
This offers a great flexibility for computing the global gradients robustly. 
Our regression training protocols are developed using an additive homomorphic encryption scheme and a secure aggregation protocol built using practical crypto-primitives.   
We also show how to obliviously compute regression models for prediction services for the future use of the trained models by the users.

\par\noindent{\bf Security.} We prove the security of the training protocol, in three different threat models, against semi-honest adversaries in the simulation paradigm. 
As our privacy-preserving training protocol is built upon several subprotocols, we first prove the security of the shared local gradient computation protocol for linear and logistic regression models. 
We  formally show that the security of the training protocol is based on the semantic security of the additive HE scheme that guarantees the server's model privacy  and the aggregation privacy game, as defined in \cite{TwoIsNotEnough}, which ensures users' data privacy.

\par\noindent{\bf Experimental  evaluation.} We implement and evaluate the efficiency of \regsys\ for training linear and logistic regression models on eleven real-world datasets from the UCI ML repository \cite{CreditDataset}. 
In our experiment, we use the Joye-Libert (JL) cryptosystem \cite{JL-cryptosystem} to realize the additive HE scheme as its ciphertext expansion is $2\times$ lesser and its ciphertext operations faster, when compared to Paillier's \cite{Paillier}. We implement an aggregation protocol that is a compilation of the protocols of \cite{SecureAggrePPML} and \cite{secure-agg-cacr2018}. 
As machine learning algorithms work on  floating-point numbers, we show how to encode floating-point numbers for cryptographic operations, and how to decode results obtained after applying crypto-operations.  
For training linear (resp. logistic) regression models, we perform experiments on 6 (resp. 5) different datasets of various sizes and dimensions. 
The accuracy of the trained models obtained by our protocols is very close to that of the model trained using \texttt{sklearn} (no security) \cite{sklearn}. 
We present the benchmark results for the execution time, data transfer, and storage overhead for cryptographic operations to train regression models.  
We provide a comparison of our scheme with other approaches in Section~\ref{sec:RelatedWork}. 

\section{Preliminaries}
Here, we briefly describe the regression algorithms, namely linear, and logistic regressions, federated learning, and the cryptographic schemes and protocols that we use to build our new protocols.
\vspace{0.1cm}
\par\noindent{\bf Basic notations.} We denote the message space by $\Z_{2^k}$, a ring of $2^k$ elements, and $\E$ by an additive HE encryption scheme. $\xbu = (x_1, \cdots, x_n), x_i \in \Z_{2^k}$ denotes an $n$ dimensional vector over $\Z_{2^k}$. 
\begin{itemize}
 \item For $\xbu = (x_1, \cdots, x_n)$, $\E(\xbu) = (\E(x_1), \cdots, \E(x_n))$. 
 \item $\E(\xbu)\cdot \E(\ybu) = (\E(x_1)\cdot\E(y_1), \cdots, \E(x_n)\cdot\E(y_n))$
  \item $\E(\xbu)^{\ybu} = (\E(x_1y_1), \cdots, \E(x_ny_n))$
 \item For $\thetabu = (\theta_0, \theta_1, \cdots, \theta_n)$, $\thetabu\cdot\xbu = \theta_0 + \sum_{i = 1}^n \theta_i x_i$. 
\end{itemize}

\subsection{Overview of Regression Algorithms}\label{sec:RegressionAlgos}
We now describe three regression algorithms and the gradient descent algorithm for the training process. 
Let $\DS = \big\{ (\xbu^{(i)}, y^{(i)}) \big\}_{i = 1}^d$ be a training dataset with labeled output $y^{(i)} = h(\thetabu, \xbu^{(i)})$ and input $\xbu^{(i)} = (x^{(i)}_1, x^{(i)}_2, \cdots, x^{(i)}_{n})$, where $h$ is the regression algorithm, $\thetabu$ is the model, and $n$ is the dimension of $\xbu$.  
The task in the training process is to learn the model $\thetabu$, given $\DS$ and $h$. 

\par\noindent{\bf Gradient Descent Algorithm.} Gradient descent (GD) is an iterative optimization algorithm for minimizing a cost function. Given a dataset $\DS$, the cost function is defined as
$J(\thetabu) = \frac{1}{|\mathcal{B}|} \sum_{(\xbu, y) \in \mathcal{B}} C(\thetabu, (\xbu, y) )$ where $C(,\cdot,)$ is a cost function and $\mathcal{B} \subseteq \DS$. An optimal $\thetabu$ is computed iteratively as 
$$\thetabu^{i+1} \leftarrow \thetabu^{i} -\eta \nabla J(\thetabu^{i}), i \geq 0$$
where $\thetabu^{0}$ is initialized with a random value or all-zero, $\eta$ is the learning rate, and $\nabla J$ is the gradient of $J$ over $\mathcal{B}$. 

\par\noindent{\bf Federated Learning.} 
Federated learning enables a large number of users to collaboratively learn a model while keeping all training data on their devices where the model is trained under the coordination of a central server \cite{FederatedLearning,McMahan-model-averaging-deep-learning}. 
We assume that there are $m$ users, denoted by $P_i$, with a set of data points, denoted by $\DS_i$. The cost function for the joint dataset $\DS = \cup_{i = 1}^{m} \DS_i$ is given by  
$J(\thetabu) = \sum_{i = 1}^m \frac{d_i}{d}F_i(\thetabu)$ 
where $F_i(\thetabu) = \frac{1}{d_i} \sum_{(\xbu,y) \in \DS_i} C(\thetabu, (\xbu, y))$ and $d = \sum_{i = 1}^m d_i$. The model update on the total training set $\DS$ is performed as \cite{McMahan-model-averaging-deep-learning}
$$\thetabu^{i+1} \leftarrow \thetabu^{i} -\eta \sum_{j = 1}^m \frac{d_j}{d} \nabla F_j(\thetabu^{i}), i \geq 0$$
where $\nabla F_j$ is the {\it local} gradient value on $\DS_j$. In federated averaging \cite{McMahan-model-averaging-deep-learning}, the users are chosen randomly. 

\par\noindent{\bf Linear Regression.} 
For a linear regression model, the function $h(\cdot, \cdot)$ is an affine function, which is defined as $h(\thetabu, \xbu^{(i)}) = \thetabu\cdot\xbu^{(i)} = \theta_0 + \sum_{j = 1}^{n}\theta_j x_j^{(i)}$. 
The model parameter $\thetabu$ is obtained by optimizing the following cost function  
$$J(\thetabu) = \frac{1}{d}\sum_{i = 1}^d ( h(\thetabu, \xbu^{(i)})  - y^{(i)})^2.$$
\vspace{-0.15cm}
\par\noindent{\bf Ridge Regression.}
For the ridge regression, the function $h$ is the same as the function of linear regression, but the model parameter is obtained by optimizing the following cost function 
$$J(\thetabu) =  \frac{1}{d} \sum_{i = 1}^{d} (h(\thetabu, \xbu^{(i)}) - y^{(i)})^2 + \lambda || \thetabu||^2$$
where $\lambda$ is the regularization parameter.

\par\noindent{\bf Logistic Regression.}
The logistic regression is a binary linear classifier, which maps an input $\xbu^{(i)}$ from the feature space to a value in $[0,1]$ as follows: 
$$h(\thetabu, \xbu^{(i)}) = \sigma(\thetabu \cdot \xbu^{(i)}) = \frac{1}{1 + e^{-(\sum_{ i = j}^n \theta_i x_j^{(i)} + \theta_0)}}$$  
where $\sigma(z) = \frac{1}{1+e^{-z}}$ is the sigmoid function.  The binary class is decided based on a threshold value. 

For the logistic regression, the model parameter $\thetabu$ is obtained by optimizing the following log-likelihood cost function  
$$J(\thetabu) = -\frac{1}{d}\sum_{i = 1}^d  y^{(i)}\log( h(\thetabu, \xbu^{(i)}) ) + (1-y^{(i)})\log( 1- h(\thetabu, \xbu^{(i)}) ).$$

\subsection{Multiparty Federated Regression Training and Oblivious  Prediction}\label{sec:muliparty-fed-learn}
\par\noindent{\bf Multiparty Gradient Computation.} 
Assume that there are $m$ users, and each user has a dataset $\DS_i, 1 \leq i \leq m$. A server wishes to learn a regression model  $\thetabu$ on the joint dataset $\mathcal{D} = \DS_1 \cup \cdots \cup \DS_m = \big\{(\xbu^{(i)}, y^{(i)})\big\}$ (horizontally distributed). 
To train a model based on the gradient descent algorithm, the server jointly computes the gradient $\omegabu$ on $\mathcal{D}$ as 
$$\omegabu  = (\omega_0, \omega_1, \cdots, \omega_n)= \nabla_{\thetabu} J(\thetabu) = \gd( \thetabu, \DS, d)$$
where $d = |\DS|$, $\omega_0 = \sum_{i = 1}^d e^{(i)}$, $ \omega_j = \sum_{i = 1}^d e^{(i)}x_{j}^{(i)}, 1 \leq j \leq n$, $e^{(i)} = h(\thetabu, \xbu^{(i)})  - y^{(i)}$ and $h$ is a (linear, ridge, or logistic) regression function.  
We denote the gradient computation on a single data point $(\xbu^{(i)}, y^{(i)})$ by $\gdbu^{(i)} = f( \thetabu, (\xbu^{(i)}, y^{(i)}) ) = (e^{(i)}, e^{(i)}x_1^{(i)}, \cdots, e^{(i)}x_n^{(i)})$ with $e^{(i)} = h(\thetabu, \xbu^{(i)})  - y^{(i)}$. 
Knowing the model $\thetabu$, each user can compute the local gradient $\omegabu_i$ on $\DS_i$ as $\omegabu_i = \sum_{ (\xbu, y) \in \DS_i } \gdbu^{(i)} =  \sum_{ (\xbu, y) \in \DS_i }  f( \thetabu, (\xbu, y) )$. 
Thus, the global gradient can be written as $\omegabu = \sum_{i = 1}^m \omegabu_i =   \sum_{j = 1}^d \gdbu^{(j)}$.  Using the gradient $\omegabu$ over $\DS$, the server can update the model as $\thetabu \leftarrow \thetabu - \frac{\eta}{d} \cdot \omegabu$ for linear and logistic regressions, and $\thetabu \leftarrow (1-2\lambda \eta)\thetabu - 2\eta \cdot \omegabu$ for ridge regression  for some regularization parameter $\lambda$.
Secure gradient computation enables the server to learn only $\omegabu$, while ensuring users' input ($\DS_i$) privacy against the server, and the server's model privacy against the users. 

\par\noindent{\bf Oblivious Regression Prediction.} Oblivious prediction for neural networks was introduced in \cite{liu-onn-ccs-2017}. An oblivious prediction for regression models is defined as follows. Assume that the server holds a (private) regression model $\thetabu$, and a user has a private input $\xbu$ and wishes to learn the predicted value $h(\thetabu, \xbu)$ for a regression function $h$ with model $\thetabu$. 
In an oblivious prediction computation, the user will learn only the predicted value  $h(\thetabu, \xbu)$ without leaking any information about $\xbu$ to the server, and the user should not learn any information about $\thetabu$, except what can be learnt from the output. 

\subsection{Cryptographic Tools} 
We now provide a description of the cryptographic schemes that we will use to construct our protocols. 
\par\noindent{\bf Additive Homomorphic Encryption.}
An additive homomorphic encryption (HE) scheme consists of a tuple of four algorithms, denoted as $\textsf{HE} = (\textsf{KeyGen}, \textsf{Enc}, \textsf{Eval}, \textsf{Dec})$ where
\begin{itemize}
\item \textsf{KeyGen}: Given a security parameter $\nu$, it generates a pair of private and public keys $(pk, sk) \leftarrow \textsf{KeyGen}(1^{\nu})$. 
\item \textsf{Enc}: It is a randomize algorithm, denoted by $\E$, takes the public key $pk$, a random coin $r$ and the message $m \in \Z_{2^k}$ as input, and outputs a ciphertext $c = \E(pk; m, r)$. 
\item \textsf{Eval}: It supports the following operations on ciphertexts: 1) Addition:  $\textsf{Add}(\E(x), \E(y)) = \E(x+y)$; and 2) Constant multiplication: $\textsf{ConstMul}(\E(x), z) = \E(xz)$.
\item \textsf{Dec}: The algorithm, denoted by $\D$, takes the secret key $sk$ and a ciphertext $c$ as input, and outputs a plaintext message $m = \D(sk; c)$. 
\end{itemize} 
We use an additive HE scheme that is semantically  secure, and $\textsf{Add}(,)$ and $\textsf{ConstMul}(,)$ are ciphertext multiplication and exponentiation operations, respectively.

\par\noindent{\bf Secure Aggregation Protocol.} 
We will use a secure aggregation protocol that can handle the dropout scenario as well as no dropout in an aggregated-sum computation. 
We call this protocol a {\it dropout-enabled aggregation} protocol, denoted by $\pi_{\textsc{DeA}}$. A dropout-enabled aggregation protocol accepts as input a set of users $\us$, 
their private inputs $\{x_u\}_{u \in \us}$, the total number of users $m$, and the threshold security parameter $t$, and outputs $x_{sum}$, and a set of {\it alive} users $\us_{a} \subseteq \us$ participated in the sum computation, i.e., 
$$(\us_a, x_{sum}) \leftarrow \pi_{\textsc{DeA}} \big(\us, \{x_u\}_{u \in \us}, m, t\big)$$ where $x_{sum} = \sum_{u \in \us_{a}} x_{u}$ if $|\us_{a}| \geq t$, and $\perp$ otherwise. 
Note that the aggregation scheme is secure under a coalition of up to $(t-1)$ users in the system. 
Such protocols with the dropout-enabled property have been investigate recently in \cite{SecureAggrePPML,secure-agg-cacr2018}. For more details, see \cite{SecureAggrePPML,secure-agg-cacr2018}. 
These aggregation protocols need to establish a pair of pairwise keys: one key will be used for realizing an authenticated channel, and another key will be used to encrypt private inputs. 
We provide a brief background on the cryptographic primitives that are needed to construct the $\pi_{\textsc{DeA}}$ protocol in Appendix~\ref{sec:CryptoBkgd}. 

\section{Our System Model and Goals}\label{sec:SysModel}
In this section, we describe the system model, its goals, and possible privacy leakage in the existing multiparty gradient computation.
\subsection{System Model and Trust}
\par\noindent{\bf System model.} 
We consider a system in the federated learning setting introduced in \cite{FederatedLearning,TowardsFL}.  
In this model, the system consists of two types of parties: a server and a set of mobile users or parties connected in mobile network, where the server conducts the regression training  process on mobile users' data.   
The users may {\it dropout} any time from the system, as recently considered in \cite{SecureAggrePPML,secure-agg-cacr2018}.

Assume that there are $m$ mobile users in the system, and each user has a unique identity in $[m] = \{1, 2, \cdots, m\}$. We denote the server by $S$, and a user by $P_i, i \in [m]$. 
Each user $P_i$ holds a dataset with high-dimensional points, denoted by $\mathcal{D}_i = \{(\xbu^{(j)}, y^{(j)})\}$ where $|\mathcal{D}_i| \geq 1$. We have the following system goals:
\begin{itemize}
\item {\bf Dropout-robust secure regression training:} The first goal of our system is to enable the server training a regression model $\thetabu$ (e.g., linear, ridge or logistic) over the combined dataset $\DS = \DS_1 \cup \cdots \cup \DS_m$, i.e., $\thetabu \leftarrow \textsf{Traning}( h , \DS, m)$, while allowing user dropout any time in the system, and  ensuring mobile users' input privacy and the servers' model privacy where $h$ is the regression algorithm (see Section~\ref{sec:muliparty-fed-learn}).  
\item{\bf Oblivious prediction:} The second goal is to enable a mobile user learn the predicted output $ y = h(\thetabu, \xbu)$ on its private input $\xbu$, corresponding to a regression model $\thetabu$, without leaking any information about $\xbu$ to the server and about $\thetabu$ to the user. 
\end{itemize}
Figure~\ref{fig:SystemOverview} provides an overview of our system and approach to achieve system goals. The main challenges in the training phase are: 
\begin{itemize}
 \item {\it Correctness:} For correct inputs of the users, the protocol for the regression training should output the correct model. We claim no correctness of  the regression model if any user uses an incorrect input or the server manipulates the model in the training phase. 
 \item {\it Privacy:} Our system has aimed at protecting users' inputs privacy and the server's model privacy. The server should learn no information about mobile users' private inputs in $\DS_i$. Similarly, the users should not learn anything about the model $\thetabu$, except what they can infer from the output. 
 \item {\it Efficiency:} As mobile users do not send  the private inputs out of the devices, they should perform a minimal work, and the server should perform the majority of work in the training phase. 
 The computational and communication costs of the training protocol should be minimal. 
\end{itemize}
\begin{figure*}[h]
\vspace{-0.35cm}
\centering
\begin{tikzpicture}
 \node at (-4,0) {\includegraphics[width=.1\textwidth]{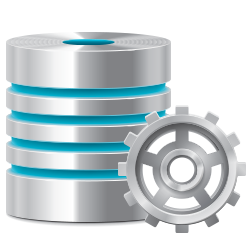}};
  \node at (-4.2,1.1) {\includegraphics[width=.06\textwidth]{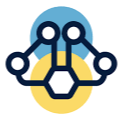}};

  \node at (-1,-3) {\includegraphics[width=.04\textwidth]{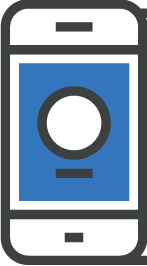}};
   \node at (-2.5,-3) {\includegraphics[width=.04\textwidth]{figs/mobile.png}};
   \node at (-4.5,-3) {$\bigcdot\bigcdot\bigcdot$};
  \node at (-6.5,-3) {\includegraphics[width=.04\textwidth]{figs/mobile.png}};
   \node at (-8,-3) {\includegraphics[width=.04\textwidth]{figs/mobile.png}};
   
    \node at (-9.0,-2.75) { {\bf User}};
   \node at (-8.5,-3.0) {\small $\rbu_1$};
   \node at (-7,-3.0) {\small $\rbu_2$};
   \node at (-3.2,-3.0) {\small $\rbu_{m-1}$};
   \node at (-1.6,-3.0) {\small $\rbu_{m}$};
   
   \node at (-8,-3.8) {\small $\DS_1$};
   \node at (-6.5,-3.8) {\small $\DS_2$};
   \node at (-2.4,-3.8) {\small $\DS_{m-1}$};
   \node at (-1,-3.8) {\small $\DS_m$};
   
   \node at (-8.2,-0.9) {\small $\DS = \DS_1 \cup \cdots \cup \DS_m$};
   
   \draw[<->, thick] (-5,-0.75) -- (-8, -2.25);
   \draw[<->, thick] (-4.5,-0.95) -- (-6.5, -2.25);
   
   \draw[<->, thick] (-3,-0.85) -- (-1, -2.25);
   \draw[<->, thick] (-3.5,-0.95) -- (-2.5, -2.25);
   
    \node at (-6.5,-1.2) { $\E(\thetabu)$};
    \node at (-7.1,-2.1) { $E(\sbu_1)$};
    
    \node at (-5.3,-1.2) { $\E(\thetabu)$};
    \node at (-3.8,-1.2) { $\E(\thetabu)$};
    \node at (-2.9,-1.2) { $\E(\thetabu)$};
    
    \node at (-5.75,-2.1) { $E(\sbu_2)$};
    \node at (-1.975,-2) { $E(\sbu_{m-1})$};
    \node at (-0.7,-2) { $E(\sbu_m)$};
    
    \node at (-4,-2.25) { $\sbu_i - \rbu_i = \omegabu_i$};
    
    \node at (-7.5,-4.1) {\small\bf (1) Training regression models};
   
   \node[draw,dashed,text width=2.5cm] at (-1.9,0.5) {Update model:\\$\thetabu \leftarrow \thetabu -\frac{\eta}{d} \sum_{i = 1}^{d}\omegabu_i$};
   \node  at (-1.5,-0.5) {$\omegabu = \sum \sbu_i - \sum \rbu_i$};
   \node at (-3,1.3) { {\bf Server}};
   
   \node[draw,dashed,text width=4.3cm] at (-7.3,0.5) {\begin{itemize} \item[1)] Establish pairwise key \item[2)] Scaling the dataset $\DS$ \item[3)] Training the model $\thetabu$ \begin{itemize} \item[-] Send enc. model $\E(\thetabu)$ \item[-] Securely compute $\omegabu$ \end{itemize} \end{itemize}};
 
  
  \draw[thick, dashed] (-9.7,-4.3) -- (0, -4.3) -- (0, 1.7) -- (-9.7, 1.7) -- (-9.7,-4.2);
 
  \node at (3.5,0) {\includegraphics[width=.1\textwidth]{figs/server.png}};
  \node at (3.2,1.2) {\includegraphics[width=.06\textwidth]{figs/ml.png}};
  \node at (3.5,-3) {\includegraphics[width=.04\textwidth]{figs/mobile.png}};
  
  \draw[->, thick] (3.3,-2.2) -- (3.3, -0.9);
  \draw[<-, thick] (3.6,-2.2) -- (3.6, -0.9); 
  \node at (2,1.2) { Model: $\thetabu$};
  \node[draw,dashed,text width=2.2cm] at (5,1.1) {Evaluating on $\xbu$:\\$y = h(\thetabu, \xbu)$};
  \node at (3,-3) { $\xbu$};
  \node at (2.9,-1.5) { $\E(\xbu)$};
  \node at (4,-1.5) { $\E(y)$};

  \draw[thick, dashed] (0.2, -4.3) -- (6.3, -4.3) -- (6.3, 1.7) -- (0.2, 1.7) -- (0.2, -4.2);
  \node at (2.7,-3.9) { \small\bf (2) Model/Prediction as a service};
\end{tikzpicture}
\vspace{-0.25cm}
\caption{An overview of the regression training process and prediction as a service over a mobile network in the federated setting. For the training phase, each user holds a dataset $\DS_i$, and the server holds an ML model $\thetabu$. Users and the server jointly compute the global gradient $\omegabu$. For the prediction service, the user holds a data point and the server holds a trained model.}
\label{fig:SystemOverview}
\vspace{-0.25cm}
\end{figure*}


\par\noindent{\bf Threat model.} 
In our system, we consider semi-honest adversaries (inside adversary) where a group of mobile users and/or the server compromised by an adversary follow and observe the prescribed actions of the protocol and aim at learning  unintended information about $\thetabu$ or honest users' $\DS_i$ from the execution  of  the  protocol. 
We consider three different threat models: 
1) users-only adversary; 2) server-only adversary; and 3) users-server adversary. 
We assume the training is conducted in a synchronous way (each iteration within a time interval of length $\Delta$), meaning all users use the correct model to compute the local gradient and the server updates the model consistently at every round of the model update. 

\subsection{Privacy Leakage and Model Privacy}\label{sec:PrivLeakage}
\par\noindent{\bf Input privacy leakage from local gradient.} 
When a user holds a single data point $(\xbu, y)$, i.e., $\DS_j = \big\{ (\xbu, y) \big\}$, 
the gradient computation on $\DS_j$, denoted by $\omegabu = (\omega_0, \omega_1, \cdots, \omega_{n})$, leaks information about the private input $\xbu$ as 
$\omega_0 = (h(\thetabu, \xbu) - y)$ and $\omega_j = (h(\xbu, \thetabu) - y)x_j = \omega_0 \cdot x_j, 1 \leq j \leq n$, where $h$ is a regression function. 
In the cases of stochastic gradient descent and federated averaging algorithm with $|\DS_i| = 1$ \cite{McMahan-model-averaging-deep-learning}, this directly allows the server to recover $x_j$ from $\omegabu$ for $\omegabu \neq {\bf 0}$. 

In the application of smart grid data aggregation,  an analysis on aggregation is performed in \cite{TwoIsNotEnough}, which experimentally finds privacy leakage in the aggregate-sum when it has a small number of inputs. 
According to the analysis of \cite{TwoIsNotEnough}, $\omega_j = \sum_{l = 1}^{d_i}(h(\thetabu, \xbu^{(l)}) - y^{(i)})x_j^{(l)}$ and $\omega_0 = \sum_{l = 1}^{d_i}(h(\thetabu, \xbu^{(l)}) - y^{(l)})$ may leak information about 
$\xbu^{(i)}$ when $d_i$ is small. In mobile applications, when mobile devices do not have enough data, it may compromise the privacy of the input dataset $\DS_i$. 
In this work, \regsys\ enables mobile users' to train a regression model without leaking their privacy, even when $|\DS_i| = 1$ and $d = \sum d_i$ is large enough to protect input privacy. 

\par\noindent{\bf Model privacy of federated learning.}  In federated learning, a substantial focus has been put on users' data privacy. 
The authors of \cite{McMahan-model-averaging-deep-learning} mentioned about achieving the model privacy using differential privacy techniques \cite{diff-privacy}. 
However, there is no such concrete proposal, and it is also not known the accuracy of the training process after applying such techniques. 
\regsys\ takes the secure MPC approach to provide the ML model privacy as well as users' data privacy.

\vspace{-0.2cm}
\section{Our Regression Protocols}
The key objective of this work is to develop secure training protocols for regression models based on cryptographic primitives suitable for mobile devices as well as robust against the user dropout scenario. 
As the training process involves a pre-processing of the training data, called the {\it scaling or normalization} of the dataset, we show how to securely perform this using an aggregation protocol while ensuring the dropout scenario. 
We start by presenting the secure scaling operation.

\subsection{Dropout-robust Secure Scaling Operation}\label{sec:scaling}
Scaling a dataset is performed by computing the mean and standard deviation of the dataset. 
Given the set of numbers $\{x_i\}$ of size $m$, the scaling is performed as $x_i \leftarrow \frac{x_i - \mu }{\sigma}$ where $\mu = \frac{1}{m}\sum_{i =1}^m x_i$ and $\sigma = \sqrt{\frac{1}{m-1}\sum_{i = 1}^m (x_i - \mu)^2}$. 
%
For a high-dimensional dataset $\mathcal{D} = \DS_1 \cup \cdots \cup \DS_m$ from users $\us$, the scaling operation is performed by scaling individual components of each data point. 
Note that scaling is performed only on $\xbu^{(i)}$ in $\DS$. 
Let $\xbu^{(i)} = \big(x_1^{(i)}, x_2^{(i)}, \cdots, x_{n}^{(i)}\big)$ be an $n$-dimensional data point. The mean of the $\xbu^{(i)}$ component of $\mathcal{D}$ of size $d$ is given by 
$\mu = (\mu_1, \cdots, \mu_{n})$ where $\mu_j = \frac{1}{d}\sum_{i = 1}^d x_j^{(i)}$, and the standard deviation is given by $\sigma = (\sigma_1, \cdots, \sigma_{n})$ where $\sigma_j = \sqrt{\frac{1}{d-1}(\sum_{i = 1}^d {x_j^{(i)}}^2 - (2d-1)\mu_j^2}), j \in [n]$. 
Secure computation of basis statistics such as mean and standard deviation has been widely investigated using secret sharing, (labeled) homomorphic encryption, and aggregation protocols, e.g., \cite{LabeledHE,SecureStat,rmind-ieee-dsc-2018,Prio-Usenix2017}. 

As we have considered a scenario of users dropping out, we perform the scaling process using a dropout-enabled aggregation protocol ($\pi_{\textsc{DeA}}$), coordinated by the server. 
Instead of running $\pi_{\textsc{DeA}}$ twice (once for computing $\mu$ and another for $\sigma$), our idea is to run $\pi_{\textsc{DeA}}$ only once on $2n$ dimensional inputs $X^{(i)}$, which 
we construct from $n$-dimensional inputs $\xbu^{(i)}$ as follows. 
Each user $P_i$ locally computes a single input $X^{(i)}$ on its dataset $\DS_i$ as  
$$X^{(i)} = \Big(\sum_{\xbu \in \DS_i} x_1, \cdots, \sum_{\xbu \in \DS_i} x_{n}, \sum_{\xbu \in \DS_i}x_1^2, \cdots, \sum_{\xbu \in \DS_i}x_{n}^2\Big)$$ 
where $\xbu = (x_1, \cdots, x_n)$. 
This can be viewed as processing multiple-data using a single execution of the protocol. 
On inputs $\{X^{(i)}, d_i\}$ from users in $\us$ and a security parameter $t$, the server receives $(\us_a, X) \leftarrow \pi_{\textsc{DeA}} ( \us, \{X^{(i)}, d_i\}, |\us|, t)$ with $|\us| \geq t$ where $X = (X_1, \cdots, X_n, X_{n+1}, \cdots, X_{2n}) = \sum_{ u \in \us_a} X^{(u)}$ if $|\us_a| \geq t$, and $\us_a$ are the alive users who participated in the scaling process. If $|\us_a| < t$, abort the protocol. 
The server computes the mean $\mubu = (\mu_1, \cdots, \mu_n)$ and standard deviation $\sigmabu = (\sigma_1, \cdots, \sigma_n)$ as $\mu_j = \frac{X_j}{d_a},$ and $\sigma_j = \sqrt{\frac{1}{d_a-1}( X_{n+j} - (2d_a-1)\mu_j^2}), j \in [n]$ and $d_a = \sum_{u \in \us_a}d_u$. The server then sends $\mubu$ and $\sigmabu$ to all users in $\us_a$ through an authenticated channel. Using $\mubu$ and $\sigmabu$, the users scale their local datasets. 

\subsection{A High-level Overview of Our Regression Training Protocols}\label{sec:ProtocolFlow}

\par\noindent{\bf Basis idea.} To train a regression model, the users and the server execute  a multiparty global gradient computation protocol where the server gives the model  and the users provide their local datasets as inputs, and the server obtains an updated model as an output. The novelty of \regsys's multiparty gradient computation is that the global gradient computation is performed by executing multiple two-party shared local gradient protocols in parallel, followed by executing a  global gradient share-reconstruction protocol realized using an aggregation protocol (see Figure~\ref{fig:OverallProtocolFlow}).  
As shown in Section~\ref{sec:SysModel},  sending the local gradient directly to the server may leak users' datasets $\DS_i$. 
Our idea is to prevent such leakage by additively secret-sharing the local gradient $ \omegabu_i$ of a user $P_i$  between the server and the user such that $\sbu_i - \rbu_i = \omegabu_i$, where the server holds the share $\sbu_i$, and the user $P_i$ holds the share $\rbu_i$. To prevent leaking the model to the users, the server encrypts the model $\thetabu$ using an additive homomorphic encryption scheme. 
In the first phase, each user computes an encrypted share $\E(\sbu_i)$ of $\omegabu_i$ from $\E(\thetabu)$ and $\DS_i$. This is computed by a protocol called, the {\it shared local gradient (SLG)} computation protocol.
Instead of sending the individual users' shares to the server, the users send their shares in aggregate. For this, the server conducts a share reconstruction process to obtain an aggregate-sum of the shares.  
This phase is called the {\it share reconstruction phase}. 
After computing $\omegabu$ on alive users' ($\us_a$) datasets from $\{\sbu_i\}$ and $\{\rbu_i\}$ as $\omegabu = \sum_{u \in \us_a} \sbu_u - \sum_{u \in \us_a} \rbu_u$, the server updates the model as $\thetabu \leftarrow \thetabu - \frac{\eta}{d_a} \omegabu$ where $d_a = \sum_{u \in \us_a}d_u$. 
Figure~\ref{fig:OverallProtocolFlow} depicts an overview of secure multiparty global gradient computation on the joint dataset $\DS$. 
\begin{figure}[h]
\centering
\vspace{-0.25cm}
\resizebox{8.5cm}{!}{
 \begin{tikzpicture}
  \draw (0,-0.5) rectangle (1.6,0);
  \draw (0,1) rectangle (1.6,1.5);
  \draw (0,2) rectangle (1.6,2.5);
  
  \node at (0.8, -0.25) {\small SLG};
  \node at (0.8, 1.25) {\small SLG};
  \node at (0.8, 2.25) {\small SLG};
  
  \node at (0.7, 0.15) {\small $\omegabu_m = \sbu_m - \rbu_m$};
  \node at (0.7, 1.7) {\small $\omegabu_2 = \sbu_2 - \rbu_2$};
  \node at (0.7, 2.7) {\small $\omegabu_1 = \sbu_1 - \rbu_1$};
  
  \node at (0.8, 0.6) {$\vdots$};
  
  \draw[->] (-0.75,-0.4) -- (0, -0.4); \node at (-1.1,-0.4){\small $\DS_m$};
  \draw[->] (-0.75,-0.15) -- (0, -0.15); \node at (-1.1,-0.15){\small $\E(\thetabu)$};
  
  \draw[->] (-0.75,1.1) -- (0, 1.1); \node at (-1.1,1.1){\small $\DS_2$};
  \draw[->] (-0.75,1.35) -- (0,1.35); \node at (-1.1,1.35){\small $\E(\thetabu)$};
  
  \draw[->] (-0.75,2.1) -- (0, 2.1); \node at (-1.1, 2.1){\small $\DS_1$}; 
  \draw[->] (-0.75,2.35) -- (0,2.35); \node at (-1.1,2.35){\small $\E(\thetabu)$};
  
  \draw[->] (1.6,-0.4) -- (2.35, -0.4); \node at (1.975,-0.65){\small $\rbu_m$};
  \draw[->] (1.6,-0.15) -- (2.35, -0.15); \node at (2,0.1){\small $\E(\sbu_m)$};
  
  \draw[->] (1.6,1.1) -- (2.35, 1.1); \node at (1.975, 0.85){\small $\rbu_2$};
  \draw[->] (1.6,1.35) -- (2.35, 1.35); \node at (1.975, 1.6){\small $\E(\sbu_2)$};
  
  \draw[->] (1.6,2.1) -- (2.35, 2.1); \node at (1.975, 1.9){\small $\rbu_1$};
  \draw[->] (1.6,2.35) -- (2.35, 2.35); \node at (1.975, 2.6){\small $\E(\sbu_1)$};

  \draw (3.5,-0.5) rectangle (5.25,0.5);
  \draw (3.5,1.5) rectangle (5.25,2.5);
  
  \node[text width=1.5cm] at (4.45, 0) {\small Share reconstruction phase};
  \node[text width=1.5cm] at (4.45, 2) {\small Decryption + aggregation};
  
  \draw[->] (4.45, 0.6) -- (4.45, 1.4);
   \node at (4.7, 1) {\small $\us_a$};
  
  \draw (5.7,0.7) rectangle (6.5,1.35); 
  \draw[->] (6.5,1) -- (7, 1);
  
  \draw[->,dashed] (2.35, -0.15) -- (2.8, -0.15) -- (2.8, 1.6) -- (3.5, 1.6);
  \draw[->,dashed] (2.35, 1.35) -- (2.7, 1.35) -- (2.7, 2.2) -- (3.5, 2.2);
   \draw[->,dashed] (2.35, 2.35) -- (3.5, 2.35);
   \node at (3.4, 2) {$\vdots$};
   
  \draw[->] (2.35, -0.4) -- (3.5, -0.4);
  \draw[->] (2.35, 1.1) -- (3, 1.1) -- (3, 0.25) -- (3.5, 0.25);
  \draw[->] (2.35, 2.1) -- (3.2, 2.1) -- (3.2, 0.4) -- (3.5, 0.4);
  
  \node at (3.4, 0) {$\vdots$};
  
  \draw[->] (5.25, 0) -- (5.5, 0) -- (5.5, 0.85) -- (5.7, 0.85);
  \draw[->] (5.25, 2) -- (5.5, 2) -- (5.5, 1.1) -- (5.7, 1.1);
  
  \node at (6.3,1.6){\small $\omegabu = \sum_{i} \omegabu_i$};
   \node at (6.4,0.4){\small $\thetabu \leftarrow \thetabu - \frac{\eta}{d_a} \omegabu$};
   \node at (6.75,1.2){\small $\thetabu$};
    \node at (6,1.1){\small $\omegabu = $};
    \node at (6.2,0.9){\small $\sbu - \rbu$};
  
  \node at (5.9,2.2){\small $\sbu = \sum_{i} \sbu_i$};
  
  \node at (5.9,-0.2){\small $\rbu = \sum_{i} \rbu_i$};
  
  
  \draw[dashed] (-1.5,-1.2) -- (2.6,-1.2) -- (2.6,3) -- (-1.5,3) -- (-1.5,-1.2);
  \draw[dashed] (3.3,-1.2) -- (7.3,-1.2) --  (7.3,3) -- (3.3,3) -- (3.3,-1.2);
  
  \node at (0.55,-0.95){\small {\bf Computing local gradient shares}};
  \node at (5.45,-0.7){\small {\bf Computing global gradient}};
  \node at (5.45,-0.95){\small {\bf from local shares}};
  
 \end{tikzpicture}}
\vspace{-0.7cm}
\caption{\regsys's protocol flow of privacy-preserving global gradient computation for regression models.}
\label{fig:OverallProtocolFlow}
\vspace{-0.35cm}
\end{figure}
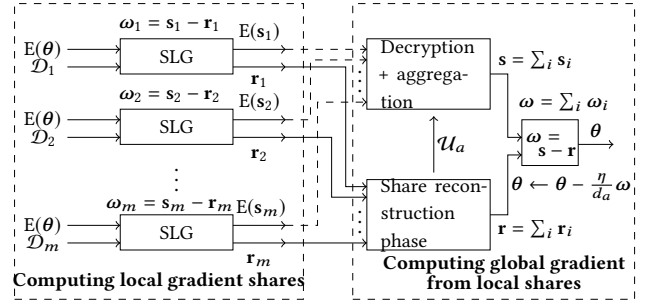
\par\noindent{\bf Advantage.} The advantages of the above flow of the protocol are as follows. 
At the $i$-th round of the training process, the users participating in the global gradient computation process do not need to know the other participating users ahead of the time.
The information about the other participating users will be enough to know only in the share reconstruction phase. This provides a great flexibility in the global gradient computation phase for 
handling the dropout scenario in mobile applications. Moreover, the server can choose the participating users randomly based on their aliveness in the global gradient computation phase. 
Note that the server can choose the participating users and let all users know ahead of time, but this will be computationally expensive for both users and the server for handling the dropout scenario. 
Notice that any two-party protocol can be used to realize the SLG computation. However, we use an additive-HE based technique for regression models to make it communication efficient because of our application scenario. 

\subsection{Secure Shared Local Gradient Computation Protocols}\label{sec:PPGD}
A secure shared local gradient computation involves a user and the server. 
We present two protocols for computing the shared local gradient computation: one protocol is for the linear regression (see Figure~\ref{fig:LinRegGDShareProtocol}), and another is for the logistic regression (see Figure~\ref{fig:LogRegGDShareProtocol}).

\par\noindent{\bf Computing the SLG for linear regression.}
To minimize the number of rounds, the server encrypts the model using an additive HE scheme and sends it to the user $P_i$ so that it can compute two shares of the local gradient $\omegabu_i$ on its dataset $\DS_i$ in one round of communication. Note that only the server holds the private key of the HE scheme.
An inner product computation of two vectors is a common task in training and evaluating a linear or logistic regression model. 
Secure inner product computation has been studied in many settings, e.g., \cite{DuAtallah,SecureScalarProduct}. 
For the sake of completeness, we provide a homomorphic inner product computation protocol from \cite{SecureScalarProduct} in Algorithm~\ref{algo:IPcomputation}. 
Figure~\ref{fig:LinRegGDShareProtocol} presents the shared local gradient computation protocol for a linear/ridge regression model. 
%
\begin{algorithm}[h]
\caption{Encrypted Inner Product}\label{algo:IPcomputation}
\begin{algorithmic}[1]
\State \textsc{Input:} $\xbu$ and $\E(\thetabu)$
\State \textsc{Output:} $\E(\thetabu\cdot \xbu)$
\Procedure{IP\_User}{}
\State $\E(\thetabu\cdot\xbu) = \prod_{j= 1}^n \E(\theta_j)^{x_j}\cdot \E(\theta_0)$
\State \Return $\E(\thetabu\cdot\xbu)$
\EndProcedure
\end{algorithmic}
\end{algorithm}
\begin{algorithm}[h]
    \centering
    \caption{Computing Shares of Local Gradient for Linear Regression}\label{algo:EncAddShareLinearReg}
    \footnotesize
   \begin{algorithmic}[1]
\State \textsc{Input:} $\mathcal{D} = \{(\xbu^{(i)}, y^{(i)})\}_{i = 1}^p$ and $\E(\thetabu)$
\State \textsc{Output:} ($\E(\sbu)$, $\rbu$) s.t. $\sbu - \rbu = \sum_{j = 1}^p \gdbu^{(j)}$
\Procedure{Comp\_LinSLG\_Share}{}
\State Randomly generate $\rbu = ( r_0, r_1, \cdots, r_n)$
\State Set $\tbu \leftarrow (0, 0, \cdots, 0)$
\For{$i = 1$ to $p$} 
\State $\E(e^{(i)}) = \textsc{IP\_User}(\xbu^{(i)}, \E(\thetabu)) \cdot \E(-y^{(i)})$
\State $\E(t_0) \leftarrow \E(t_0)\cdot \E(e^{(i)}) = \E(t_0 + e^{(i)})$
\For{$j = 1$ to $n$}
\State $\E(t_j) \leftarrow \E(t_j)\cdot\E(e^{(i)})^{ x_j^{(i)} } = \E( t_j + e^{(i)}x_j^{(i)} ) $
\EndFor
\EndFor
\State Compute $\E(\sbu) = \E(\tbu)\cdot\E(\rbu) = \E(\tbu + \rbu)$
\State \Return ($\E(\sbu)$, $\rbu$)
\EndProcedure
\end{algorithmic}
\end{algorithm}

\begin{figure}[h]
\begin{center}
\fbox{
\procedure[syntaxhighlight=auto,space=auto,width=7.5cm, mode=text]{\textbf{Secure SLG Protocol: Linear regression} ($\pi_{\textsc{LinSLG}}$ )}{
{\bf Server: } Regression model $\thetabu$, {\bf User ($P_i$): } $\DS_i$ \\
{\bf Output:} Server receives $\E(\sbu_i)$, $P_i$ receives $\rbu_i$.\\
[][\hline] \\
\vspace{-0.75cm}
1. Server encrypts the model $\thetabu$ and sends $\E(\thetabu)$ to the user $P_i$\\ 
2. User  $P_i$ runs Algorithm~\ref{algo:EncAddShareLinearReg} on inputs $\E(\thetabu)$ and $\DS_i$, and obtains $(\E(\sbu^{i}), \rbu^{i})$  
such that $\sbu_i - \rbu_i = \omegabu_i = \texttt{GD}(\thetabu, \DS_i, d_i)$. \\
3. $P_i$ stores $\rbu_i$ and sends $\E(\sbu_i)$ to the server.
}}
\vspace{-0.25cm}
\caption{Secure shared local gradient computation protocol for a linear regression model.} 
\label{fig:LinRegGDShareProtocol}
\end{center}
\vspace{-0.35cm}
\end{figure}

\par\noindent{\bf Computing the SLG for logistic regression.}
The computation of the local gradient involves an evaluation of the sigmoid function on $\thetabu\cdot\xbu^{(j)}$.
As shown for linear regression, the server and the user can compute $\E(\thetabu\cdot\xbu^{(j)})$ in one round of communication.  
As we used an additive HE scheme, the user and the server need one more round of communication to compute $\E(\sigma_3( \thetabu\cdot\xbu^{(j)})), \xbu^{(j)}\in \DS_i$ from $\E(\thetabu\cdot\xbu^{(j)})$ 
where the sigmoid function is approximated by a cubic polynomial, which provides a good tradeoff between the accuracy and the efficiency.
Let $\sigma_3(x) = q_0 + q_1 x + q_2 x^2 + q_3 x^3$ be a cubic approximation of $\sigma(x)$ where the coefficients $q_i$ are public. 
To protect an input $\xbu$ against the server, the user masks $y = \thetabu\cdot\xbu$ as $z = y + r$ by choosing a random value $r$, and then sends it to the server. The computation of $\sigma_3(y)$ can be expressed  as 
$$\sigma_3(y) = \sigma_3(z)  - \sigma_3(r) - (q_0 + 3q_3r^3) - 3q_3rz^2 - (2q_2r  - 6q_3r^2)y.$$ 
To reduce computational cost for the users, the server computes $\E(z^2)$ and $\E(\sigma_3(z))$ from $\E(z)$ and sends ($\E(z^2), \E(\sigma_3(z))$) to the user. 
Given $\E(y), \E(z^2), \E(\sigma_3(z))$ and $r$,  the user can compute $\E(\sigma_3(y))$ using the homomorphic property of the encryption scheme as
$$\E(z^2)^{- 3q_3r}\E(y)^{- (2q_2r - 6q_3r^2)} \E(-(\sigma_3(r)+(q_0 + 3q_3r^3))) \E(\sigma_3(z)).$$ 
Once the user has $\{\E(\sigma_3(\thetabu\cdot\xbu^{(j)}))\}$, it can compute the local gradient $\omegabu_i$ on $\DS_i$, using the steps described in Algorithm~\ref{algo:EncAddShareLogisticReg}. 
Figure~\ref{fig:LogRegGDShareProtocol} summarizes the shared local gradient computation protocol for a logistic regression model. 
\begin{algorithm}[H]
    \centering
    \caption{Computing Shares of Local Gradient for Logistic Regression}\label{algo:EncAddShareLogisticReg}
    \footnotesize
   \begin{algorithmic}[1]
\State \textsc{Input:} $\mathcal{D} = \{(\xbu^{(i)}, y_i)\}_{i = 1}^p$ and $\{\E(\sigma_3(\thetabu\cdot\xbu^{(i)}))\}$
\State \textsc{Output:} ($\E(\sbu)$, $\rbu$) s.t. $\sbu - \rbu = \sum_{j = 1}^p \gdbu^{(j)}$
\Procedure{Comp\_LogSLG\_Share}{}
\State Randomly generate $\rbu = ( r_0, r_1, \cdots, r_n)$
\State Set $\tbu \leftarrow (0, 0, \cdots, 0)$
\For{$i = 1$ to $p$} 
\State $\E(e^{(i)}) = \E(\sigma_3(\thetabu\cdot\xbu^{(i)}))\cdot \E(-y^{(i)})$
\State $\E(t_0) \leftarrow \E(t_0)\cdot \E(e^{(i)}) = \E(t_0 + e^{(i)})$
\For{$j = 1$ to $n$}
\State $\E(t_j) \leftarrow \E(t_j)\cdot\E(e^{(i)})^{ x_j^{(i)} } = \E( t_j + e^{(i)}x_j^{(i)} ) $
\EndFor
\EndFor
\State Compute $\E(\sbu) = \E(\tbu)\cdot\E(\rbu) = \E(\tbu + \rbu)$
\State \Return ($\E(\sbu)$, $\rbu$)
\EndProcedure
\end{algorithmic}
\end{algorithm}

\begin{figure}[h]
\begin{center}
\fbox{
\procedure[syntaxhighlight=auto,space=auto,width=7.5cm, mode=text]{\textbf{Secure SLG Protocol: Logistic regression} ($\pi_{\textsc{LogSLG}}$ )}{
{\bf Server: } Regression model $\thetabu$, {\bf User ($P_i$): } $\DS_i$ \\
{\bf Output:} Server receives $\E(\sbu_i)$, $P_i$ receives $\rbu_i$.\\
[][\hline] \\
1. Server encrypts the model $\thetabu$ and sends $\E(\thetabu)$ to the user $P_i$\\ 
2. $P_i$ randomly generates $\{c_j\}_{j = 1}^{d_i}$.\\ 
3. $P_i$ computes: for  j = 1 to $d_i$ do \\
Compute $\E(z_j) \leftarrow \textsc{IP\_User}( \xbu^{(j)}, \E(\thetabu))\cdot \E(c_j), \text{ where } z_j = \thetabu\cdot\xbu^{(j)}+c_j$\\
endfor\\
4. $P_i$ sends $\{\E(z_j)\}$ to the server \\
5. Server decrypts $\{\E(z_j)\}_{j = 1}^{d_i}\}$ and computes $\{\E(z_j^2), \E(h_q(z_j))\}_{j =1}^{d_i}$ and sends to $P_i$ \\
6. $P_i$ computes: 
 for  j = 1 to $d_i$ do \\
$t = \E(-(h_q(c_j)+(q_0 + 3q_3c_j^3))) \cdot \E(\thetabu\cdot\xbu^{(j)})^{- (2q_2c_j - 6q_3c_j^2)}$\\
$\E(\sigma_3(\thetabu\cdot\xbu^{(j)})) = \E(\sigma_3(z_j))  \E(z_j^2)^{- 3q_3c_j}  \cdot t$\\
endfor\\
7. $P_i$ runs Algorithm~\ref{algo:EncAddShareLogisticReg} on inputs $\{ \E(\sigma_3(\thetabu\cdot\xbu^{(j)})) \}$ and $\DS_i$, and obtains $(\E(\sbu_{i}), \rbu_{i}) $
such that $\sbu_i - \rbu_i = \omegabu_i = \texttt{GD}(\thetabu,  \DS_i, d_i)$. \\
8. $P_i$ stores $\rbu_i$ and sends $\E(\sbu_i)$ to the server.
}}
\end{center}
\vspace{-0.35cm}
\caption{Secure shared local gradient computation protocol for a logistic regression model.} 
\label{fig:LogRegGDShareProtocol}
\vspace{-0.35cm}
\end{figure}
\par\noindent{\bf Complexity.} 
We measure the computational complexity of $\pi_{\textsc{LinSLG}}$ and $\pi_{\textsc{LogSLG}}$ for computing the local gradient in terms of the number of ciphertext multiplications, the number of encryptions, and the number of homomorphic constant multiplications. Table~\ref{tab:NumberHEOperation} summarizes their exact numbers. 
Asymptotically, the overall computational complexity for both $\pi_{\textsc{LinSLG}}$ and $\pi_{\textsc{LogSLG}}$ is $O(n d_i)$. 

Let $\lambda$ be the bit length of a ciphertext. Then, the communication cost involved in computing $\omegabu_i$ for the linear regression is $2(n+1)\lambda$, which is for receiving the encrypted model and sending an encrypted share of $\omegabu_i$ to the server. For the logistic regression, the communication cost  is $(2(n+1)+3d_i)\lambda$ where $d_i = |\DS_i|$ is the size of the dataset, which is due to receiving the encrypted model, sending an encrypted share of $\omegabu_i$ and information exchange for homomorphically computing the approximated sigmoid function.  
\begin{table}[h]
\centering
\caption{Number of unit operations for each user in the gradient computation protocols}
\vspace{-0.35cm}
\label{tab:NumberHEOperation}
\resizebox{8.5cm}{!}{
\begin{tabular}{| l | c | c| c|}\hline
{\bf Protocol} & \#{\bf CT MUL} & \#{\bf Const MUL} & \#{\bf Enc} \\ \hline
Linear regression ($\pi_{\textsc{LinSLG}}$) & $2(n+1) d_i - (n+1)$ & $2n d_i$ & $d_i + (n+1)$ \\ \hline
Logistic regression ($\pi_{\textsc{LogSLG}}$) & $(2n+5)d_i - (n+1)$ & $2(n+2)d_i $ & $3d_i + (n+1)$ \\ \hline
\end{tabular}
}
\end{table}
\par\noindent{\bf Security.} 
We prove the security of two SLG computation protocols in the simulation paradigm. 
Theorems~\ref{thm:LinSLG} and~\ref{thm:LogSLG} summarize the security of the  SLG computation protocols for linear and logistic regressions.  
We show that an adversary $\A$ controlling a user $P_i$ (resp. the server) does not learn anything about $\thetabu$ (resp. $\omegabu_i$) during the execution of $\pi_{\textsc{LinSLG}}$ and similarly for 
$\pi_{\textsc{LogSLG}}$.
Due to the space limit, we present the security proofs in Appendix~\ref{sec:SLGSecProof}.
\begin{theorem}\label{thm:LinSLG}
Assume that the additive homomorphic encryption scheme $\E()$ is semantically secure. The protocol $\pi_{\textsc{\emph{LinSLG}}}$ between $P_i$ and $S$ securely computes two shares of the local gradient on the dataset $\DS_i$ in the presence of semi-honest adversaries. 
\end{theorem}
\begin{theorem}\label{thm:LogSLG}
Assume that the additive homomorphic encryption scheme $\E()$ is semantically secure.
The protocol $\pi_{\textsc{\emph{LogSLG}}}$ between $P_i$ and $S$ securely computes two shares of the local gradient $\omegabu_i$ on the dataset $\DS_i$ in the presence of semi-honest adversaries. 
\end{theorem}

\subsection{Privacy-preserving Regression Model Training Protocol}\label{sec:PPTraningProtocol}
%
The training protocol is executed between the server, and a set of $m$ users in a synchronous network and is divided into three phases: 1) a pairwise key establishment phase, 2) a dataset scaling phase and 3) computing the regression model by iteratively computing the global gradient over a subset of users' data. 
A pair of pairwise keys is established using a Diffie-Hellman (DH) key agreement protocol to realize an authenticated channel between a pair of users and to use in the aggregation protocol. 
Before starting the training process, the users need to learn the mean and standard deviation on the entire dataset $\DS$, which can be easily computed by running the aggregation protocol as shown in Section~\ref{sec:scaling}. 

The process of the global gradient computation \footnote{In the training phase, the batch size in the gradient computation varies in the range of $(t+\rho-1)\ell$ and $2t\ell$, and $|\DS_i| = \ell$.} consists of four main steps: 
\begin{itemize}
\item the server randomly chooses a set of $M$ users and broadcasts the current encrypted model to these users,
\item the server and each user privately compute two shares of the local gradient on the chosen user's dataset in parallel, 
\item the server and alive users compute a single  (aggregated) share from the alive users' shares of the local gradients, 
\item the server computes the second share of the global gradient from its local shares and  then recovers the  global gradient for updating the model.
\end{itemize}
The parameter $M$ is chosen based on the security parameter of the aggregation protocol and the $\epsilon$-privacy of the aggregate-sum (see Appendix~\ref{sec:AggPrivGame}). 
For simplicity, we  assume that each user has an equal number of data points, i.e.,  $|\DS_i| = d_i = \ell, \forall i$. 
For a coalition of $c$ users and/or the server in the $\epsilon$-private aggregation scheme, the number of dropouts, denoted by $\delta$, must satisfy the following constraint: $\ell ( M - \delta - c) \geq \epsilon$, which implies $\delta \leq (M- \rho - c)$ where $\rho = \lceil \frac{\epsilon}{\ell} \rceil$. 
To resist a coalition of up to $t$ parties including users and the server, $M$ must be at least $2t$ where $t = \lceil \frac{m}{3} \rceil$. 
For a coalition of size up to $t$ and $M = 2t$, $ \delta \leq (t - \rho)$. 
When $\ell = 1$, for such a coalition of size up to $t$ and $\delta = t$ dropouts, the number of users must participate in the training process is at least $2t + \epsilon$. The threshold value of $\pi_{\textsc{DeA}}$ in the scaling phase must be at least $2t$.
Figure~\ref{fig:LinRegTrainingProtocol} summarizes the complete protocol for training a linear or logistic regression model, given the parameters $\delta, t, \epsilon$ and $\ell$, while handling the dropout scenario.  
\begin{figure*}[!tbhp]
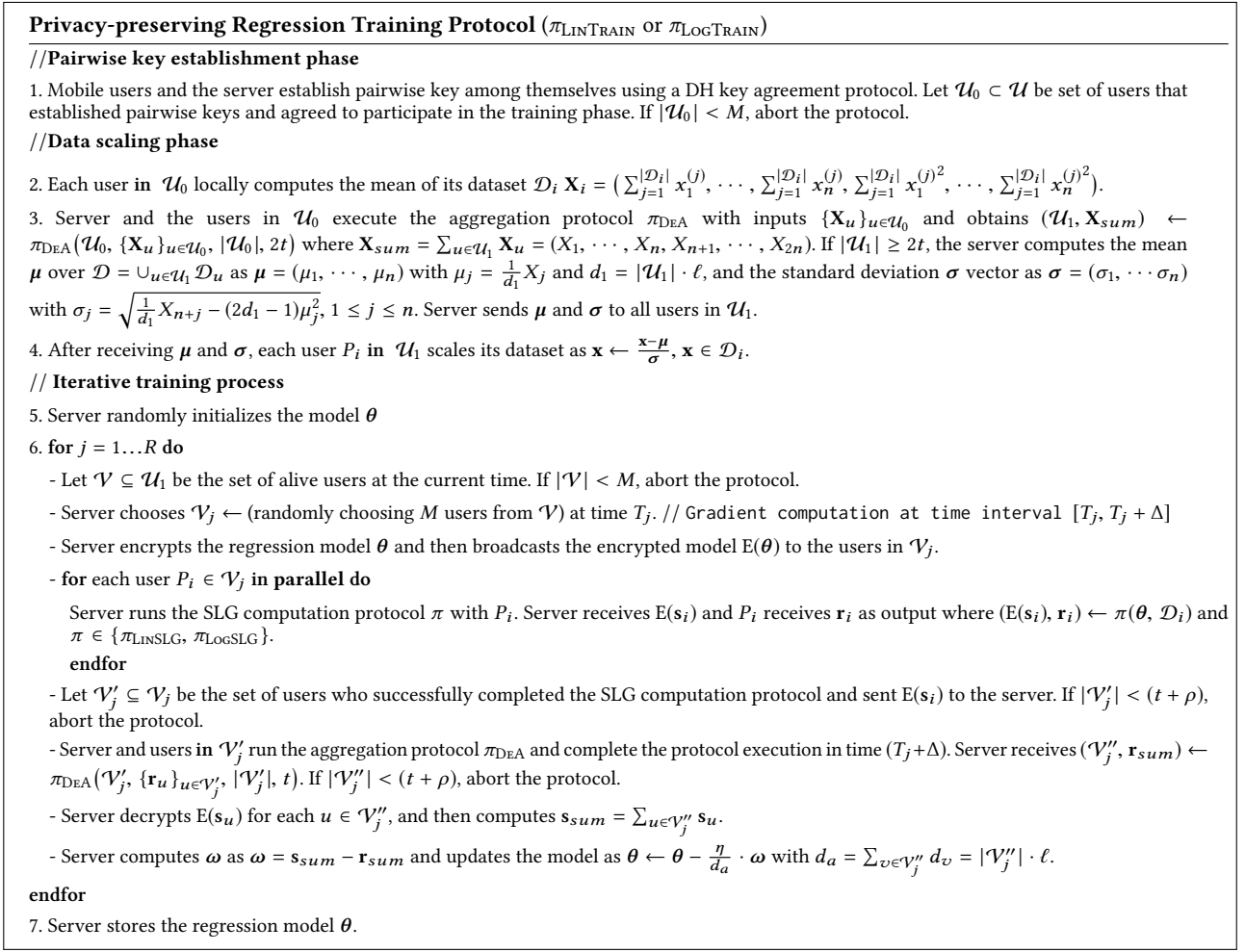

\vspace{-0.4cm}
\begin{center}
\fbox{
\procedure[syntaxhighlight=auto,space=auto,width=16cm, mode=text]{\textbf{Privacy-preserving Regression Training Protocol} ($\pi_{\textsc{LinTrain}}$ or $\pi_{\textsc{LogTrain}}$)}{
 $//${\bf Pairwise key establishment phase}\\
1. Mobile users and the server establish pairwise key among themselves using a DH key agreement protocol. Let $\us_0 \subset \us$ be set of users that established pairwise keys and agreed to participate \textrm{in} the training phase. If $|\us_0| < M$, abort the protocol.\\
 $//${\bf Data scaling phase}\\
2. Each user {\rm in } $\us_0$ locally computes the mean of its dataset $\DS_i$ $\Xbu_i = \big( \sum_{j = 1}^{|\DS_i|} x_{1}^{(j)}, \cdots,  \sum_{j = 1}^{|\DS_i|} x_{n}^{(j)}, \sum_{j = 1}^{|\DS_i|} {x_{1}^{(j)}}^2, \cdots,  \sum_{j = 1}^{|\DS_i|} {x_{n}^{(j)}}^2 \big) $. \\
3. Server and the users {\rm in} $\us_0$ execute the aggregation protocol $\pi_{\textsc{DeA}}$ with inputs $\{\Xbu_u\}_{u \in \us_0}$ and obtains $(\us_1, \Xbu_{sum}) \leftarrow \pi_{\textsc{DeA}} \big( \us_0, \{\Xbu_{u}\}_{u \in \us_0}, |\us_0| , 2t \big) $
where $\Xbu_{sum} = \sum_{u \in \us_1} \Xbu_{u} = (X_1, \cdots, X_n, X_{n+1}, \cdots, X_{2n})$.   
If $|\us_1| \geq 2t$, the server computes the mean $\mubu $ over $\DS = \cup_{u \in \us_1} \DS_{u}$ as $\mubu =  (\mu_1, \cdots, \mu_n)$ with $ \mu_j = \frac{1}{d_1}X_j$ and $d_1 = |\us_1| \cdot \ell$, 
and the standard deviation $\sigmabu$ vector as $\sigmabu = (\sigma_1, \cdots \sigma_n)$ with $\sigma_j = \sqrt{\frac{1}{d_1} X_{n+j} - (2d_1-1)\mu_j^2 }, 1 \leq j \leq n$. 
Server sends $\mubu$ and $\sigmabu$ {\rm to} all users {\rm in} $\us_1$. \\
4. After receiving $\mubu$ and $\sigmabu$, each user $P_i$ {\rm in } $\us_1$ scales its dataset as $\xbu \leftarrow \frac{\xbu - \mubu}{\sigmabu}, \xbu \in \mathcal{D}_i$. \\ 
 $//$\textbf{{ Iterative training process}}\\
5. Server randomly initializes the model $\thetabu$\\
6. for $j = 1 ... R $ do \\ 
  - Let $\vs \subseteq \us_1$ be the set of alive users at the current time. If $|\vs| < M$, abort the protocol.\\
  - Server chooses $\vs_j \leftarrow$ (randomly choosing $M$ users from $\vs$) at time $T_j$. $//$ \texttt{Gradient computation at time interval $[T_j, T_j+\Delta]$}\\
  - Server encrypts the regression model $\thetabu$ and then broadcasts the encrypted model $\E(\thetabu)$ {\rm to} the users {\rm in} $\vs_j$. \\
  - for each user $P_i \in \vs_j$ in {\bf parallel} do \\
   Server runs the SLG computation protocol $\pi$ with $P_i$. Server receives $\E(\sbu_{i})$ and $P_i$ receives $\rbu_{i}$ as output where $(\E(\sbu_{i}), \rbu_{i}) \leftarrow \pi(\thetabu, \DS_i)$ and $\pi \in \{\pi_{\textsc{LinSLG}}, \pi_{\textsc{LogSLG}}\}$.\\ 
  \t endfor\\
  - Let $\vs_j' \subseteq \vs_j$ be the set of users who successfully completed the SLG computation protocol and sent $\E(\sbu_{i})$ {\rm to} the server. If $|\vs_j'| < (t+\rho)$, abort the protocol. \\
  - Server and users in $\vs_j' $ run the aggregation protocol $\pi_{\textsc{DeA}}$ and complete the protocol execution {\rm in} time $(T_j+\Delta)$. Server receives $(\vs_j'', \rbu_{sum}) \leftarrow \pi_{\textsc{DeA}} \big(\vs_j', \{\rbu_u\}_{u \in \vs_j'}, |\vs_j'| , t\big)$. 
  If $|\vs_j''| < (t+\rho)$, abort the protocol. \\
  - Server decrypts $\E(\sbu_u)$ {\rm for} each $u \in \vs_j''$, and then computes $\sbu_{sum} = \sum_{u \in \vs_j''} \sbu_{u}$. \\
  - Server computes $\omegabu$ as $\omegabu = \sbu_{sum} - \rbu_{sum}$ and updates the model as  $\thetabu \leftarrow \thetabu - \frac{\eta}{ d_a} \cdot \omegabu$ with $d_a = \sum_{v \in \vs_j''} d_v = |\vs_j''|\cdot \ell$. \\
endfor \\
7. Server stores the regression model $\thetabu$.
}}
\vspace{-0.35cm}
\caption{Privacy-preserving training protocol for a linear or logistic regression model over a mobile network.} 
\label{fig:LinRegTrainingProtocol}
\end{center}
\vspace{-0.25cm}
\end{figure*}

\par\noindent{\bf Efficiency.} It is easy to verify that the training protocol is correct if the users provide their true inputs and the server does not alter the model. 
 Since the aggregation protocol is executed $(R+1)$ times, the users establish the pairwise keys using a Diffie-Hellman protocol, coordinated by the server, as shown in \cite{SecureAggrePPML},  
and then in each execution of $\pi_{\textsc{DeA}}$ the users derive a pair of one-time pairwise keys from the master pairwise keys using a hash-chain, as shown in \cite{secure-agg-cacr2018}. 
Thus, the key establishment phase is executed only once. 
The computational complexity of the aggregation protocol $\pi_{\textsc{DeA}}$ with $m$ users and $n$ dimensional inputs for a user is $O(m^2 + mn)$, and for the server is $O(m^2n)$. 
The computational complexity for one execution of the global gradient computation phase involves the computational complexities of $\pi_{\textsc{LinSLG}}$ or $\pi_{\textsc{LogSLG}}$ and  $\pi_{\textsc{DeA}}$, which is 
$O(4t^2 + 2tn + nd_i) = O(m^2 + mn + nd_i)$ when $t = \lceil \frac{m}{3}\rceil$. If the training phase takes $R$ iterations to obtain a model, the overall computational complexity of a user  is $O(R(m^2 + mn + nd_i))$. 
The computational complexity of the server for the training phase  is $O(R(m^2n + mnd_i))$ which includes the computational complexities of the scaling phase and $R$ executions of the privacy-preserving gradient computation phase. 

The communication complexity of each user includes the costs of the DH key exchange protocol coordinated by the server, the scaling phase and the $R$ iterations of $\pi_{\textsc{LinSLG}}$ or $\pi_{\textsc{LogSLG}}$ and $\pi_{\textsc{DeA}}$. When we choose $t = \lceil \frac{m}{3} \rceil$, the communication complexity of $\pi_{\textsc{DeA}}$,  $\pi_{\textsc{LinSLG}}$ and $\pi_{\textsc{LogSLG}}$ is $O(m+n)$, $O(n)$ and $O(2n+3d_i)$, respectively. 
Therefore, the overall communication complexity of each user is $O(R(m+n))$ for linear regression and  $O(R(m + n + d_i))$ for logistic regression, respectively. 
The server's communication complexity is $O(R(m^2+mn))$ for linear regression, and $O(m^2 + Rm(n+d_i)))$ for logistic regression, including the pairwise communication costs of protocols $\pi_{\textsc{LinSLG}}$ or $\pi_{\textsc{LogSLG}}$ and $\pi_{\textsc{DeA}}$. 

The storage overhead of the protocol is dominated by that of the aggregation protocol $\pi_{\textsc{DeA}}$. The storage overhead for each user is $O(m+n)$ for linear regression, and  $O(m+n + d_i)$ for logistic regression, which is for storing secret-shares sent by the other users in $\pi_{\textsc{DeA}}$, the encrypted model and intermediate results for logistic regression. 
As the server needs to store the model, shares of the local gradients, the secret-shares of $\pi_{\textsc{DeA}}$,  and intermediate results of  $\pi_{\textsc{LogSLG}}$, the storage complexity of the server is $O(m^2+mn)$ for linear regression, and $O(m^2+mn+md_i)$ for logistic regression.  

\par\noindent{\bf Security.} 
We consider the security of the training protocols against semi-honest adversaries where three different threat models, namely users-only threat model, server-only threat model and users-server threat model are considered.
The security of the training protocols are proved in the simulation paradigm using the hybrid arguments. 
In Theorem~\ref{thm:SecProofTrain}, we present that a semi-honest adversary corrupting at most $t$ parties including the server and a set of at most $(t-1)$ users can learn no information about the honest users' datasets. 
The security of the regression training protocols is achieved by that of the SLG protocols, the aggregation protocol and the aggregation privacy game defined in Appendix~\ref{sec:AggPrivGame}. 
Due to the space limit, a formal security proof is provided in Appendix~\ref{sec:TrainSecProof}.
\begin{theorem}[Privacy in Users-Server Threat Model]\label{thm:SecProofTrain}
The protocols $\pi_{\textsc{\emph{LinTrain}}}$ and $\pi_{\textsc{\emph{LogTrain}}}$ are secure in the presence of semi-honest adversaries, meaning they leak no information about the honest users' inputs $\DS_i$ with $|\DS_i| \geq 1$ to the adversary corrupting the server and a set of users of size up to $(t-1)$.
\end{theorem}
\vspace{-0.1cm}
In Theorem~\ref{thm:SecProofUserOnlyTrain}, we show that a semi-honest adversary corrupting at most $(t-1)$ users in the training phase learns no information about the honest server's model. 
Intuitively, since the model is encrypted using a semantically secure additive HE scheme, the model privacy is protected against semi-honest users.
\vspace{-0.2cm}
\begin{theorem}[Privacy in Users-only Threat Model]\label{thm:SecProofUserOnlyTrain}
The protocols $\pi_{\textsc{\emph{LinTrain}}}$ and $\pi_{\textsc{\emph{LogTrain}}}$ are secure in the presence of semi-honest adversaries, meaning they leak no information about the honest server's model $\thetabu$ to the adversary corrupting a set of users of size up to $(t-1)$. 
\end{theorem}

\subsection{Oblivious Regression Prediction Protocols}
We consider a scenario where, after training the model, users wish to use the trained regression model as predictions as a service, which is quite natural because the model was trained on their datasets, 
or the server wishes to offer regression predictions as a service to other clients.  
An oblivious regression prediction protocol is run between the server and a user. In this computation, the server holds a secret regression model $\thetabu$, and a user has a private input $\xbu$. 
An oblivious regression prediction protocol allows the user only to learn $h(\thetabu, \xbu)$ without revealing $\xbu$ to the server and $\thetabu$ to the users.

We assume that each user has a public and private key pair, denoted by $(pk_{P_i}, sk_{P_i})$, of an additive homomorphic encryption scheme and its public key is known to the server.
For oblivious predictions, the user sends an encrypted input to the server. For linear regression, the server computes $\E_{pk_{P_i}}(\thetabu\cdot\xbu)$ from $\E_{pk_{P_i}}(\xbu)$ and $\thetabu$. 
For logistic regression, our oblivious prediction protocol uses the same idea on a single input as used in the SLG computation protocol. 
Figures~\ref{fig:OblLinRegProtocol} and~\ref{fig:OblLogRegProtocol} present the oblivious regression prediction protocols for linear and logistic regressions, respectively. 

\par\noindent{\bf Efficiency.}
For the linear regression prediction, a user needs to perform $n$ encryptions to encrypt the input $\xbu$ of dimension $n$ and one decryption to decrypt the result. 
The server needs to perform $n$ ciphertext multiplications, $n$ homomorphic constant multiplications and one encryption. 
On the other hand, for the logistic regression prediction, the user need to perform $(n+2)$ encryptions, and  two decryptions. 
The server needs to perform $(n+3)$ ciphertext multiplications, $(n+2)$ constant multiplications and three encryptions. 

The communication cost of a user for the oblivious prediction is measured in terms of the amounts of bits the user needs to exchange. 
The amount of bits a user needs to exchange for linear regression is $(n+1)\lambda$, and for the logistic regression is $(n+4)\lambda$. 

\begin{figure}[h]
\vspace{-0.25cm}
\begin{center}
\fbox{
\procedure[syntaxhighlight=auto,space=auto,width=7.5cm, mode=text]{\textbf{Oblivious Predicion: Linear Regression} ($\pi_{\textsc{LinOP}}$)}{
{\bf Server: } Private model $\thetabu$, 
{\bf User ($P_i$): } Private input $\xbu$ \\
{\bf Output:} $P_i$ receives $h(\thetabu, \xbu)$.\\
[][\hline] \\
1. User $P_i$ encrypts its input $\xbu$ and sends $\E_{pk_{P_i}}(\xbu)$ {\rm to} the server\\ 
2. Server computes $\E_{pk_{P_i}}(h(\thetabu, \xbu)) \leftarrow \prod_{j= 1}^n \E_{pk_{P_i}}(x_j)^{\theta_j}\cdot \E_{pk_{P_i}}(\theta_0)$ where $h(\thetabu, \xbu) = \xbu\cdot \thetabu$. It sends $\E_{pk_{P_i}}(h(\thetabu, \xbu))$ {\rm to} $P_i$.\\
3. $P_i$ decrypts $\E_{pk_{P_i}}(\xbu\cdot \thetabu)$ and obtains $h(\thetabu, \xbu)$.
}}
\vspace{-0.35cm}
\caption{Oblivious linear regression prediction protocol.} 
\label{fig:OblLinRegProtocol}
\end{center}
\vspace{-0.2cm}
\end{figure}
\begin{figure}[h]
\vspace{-0.3cm}
\begin{center}
\fbox{
\procedure[syntaxhighlight=auto,space=auto,width=7.5cm, mode=text]{\textbf{Oblivious Prediction: Logistic Regression} ($\pi_{\textsc{LogOP}}$)}{
{\bf Server: } Private model $\thetabu$, 
{\bf User ($P_i$): } Private input $\xbu$ \\
{\bf Output:} $P_i$ receives $\sigma_3(\thetabu\cdot\xbu)$.\\
[][\hline] \\
1. User $P_i$ encrypts $\xbu$ and sends $\E(\xbu)$ {\rm to} the server\\ 
2. Server randomly generates $r$, and computes $\E(z) \leftarrow \prod_{j= 1}^n \E_{pk_{P_i}}(x_j)^{\theta_j}\cdot \E_{pk_{P_i}}(\theta_0)\cdot \E_{pk_{P_i}}(r), \text{ where } z = \thetabu\cdot\xbu+r$ and sends $\{\E_{pk_{P_i}}(z)\}$ to the user \\
3. $P_i$ decrypts $\E_{pk_{P_i}}(z)$ and computes $\E_{pk_{P_i}}(z^2), \E_{pk_{P_i}}(\sigma_3(z))$ and sends {\rm to} the server \\
4. Server computes the following:\\
\t\t $t = \E_{pk_{P_i}}(-(h_q(r)+(q_0 + 3q_3r^3))) \cdot \E_{pk_{P_i}}(\thetabu\cdot\xbu)^{- (2q_2r - 6q_3r^2)}$\\
\t\t $\E_{pk_{P_i}}(\sigma_3(\thetabu\cdot\xbu)) = \E_{pk_{P_i}}(\sigma_3(z))  \E_{pk_{P_i}}(z^2)^{- 3q_3r}  \cdot t$\\
5. $P_i$ decrypts $\E_{pk_{P_i}}(\sigma_3(\thetabu\cdot\xbu))$ and obtains $\sigma_3(\thetabu\cdot\xbu)$. 
}}
\end{center}
\vspace{-0.35cm}
\caption{Oblivious logistic regression prediction protocol.} 
\label{fig:OblLogRegProtocol}
\vspace{-0.45cm}
\end{figure}

\par\noindent{\bf Security.} It is shown that the attacks on the model can be performed when used as prediction services e.g., in \cite{ModelStealingAttack,ModelInversionAttacks} 
due to querying the model multiple times to access the prediction APIs. We assume that no information is leaked about the model through the prediction APIs under a bounded number of queries. 
As countermeasures proposed in \cite{ModelStealingAttack,ModelInversionAttacks} to prevent leakage, which is out of the scope of this work. 
Theorem~\ref{thm:OPResult} summarizes the security of regression prediction protocols. 
\begin{theorem}\label{thm:OPResult}
The protocols $\pi_{\textsc{\emph{LinOP}}}$ and $\pi_{\textsc{\emph{LogOP}}}$ are secure in the presence of semi-honest adversaries, meaning they leak no information about the model $\thetabu$ to a semi-honest user, and no information $\xbu$ to a semi-honest server, 
except what can be learnt obviously. 
\end{theorem}
\begin{proof}
 It is easy to see that both protocols output the correct result if the parties are honest. 
 We will first consider the privacy of the user input. In both protocols, the server observes an $n$-dimensional ciphertext vector, which is an encryption of $\xbu$ under a semantically secure HE scheme, and hence it learns nothing about the input. 
 Thus, the input privacy of the user is protected.
 
 We now consider the privacy of the model against a semi-honest user. For $\pi_{\textsc{LinOP}}$, the user observes the output $h(\thetabu, \xbu)$ after encryption, thus it learns no information about the model after one query. 
 In $\pi_{\textsc{LogOP}}$, the user observes $z$ after decrypting $\E(z)$ and learns nothing about $z$ as $z =  \thetabu\cdot\xbu+r$ is random where $r$ is chosen uniformly at random by the server. 
 Thus the user learns only $\sigma_3(\thetabu\cdot\xbu)$. 
\end{proof}

\section{Experimental Evaluation}
In this section, we evaluate the performance of \regsys\ under semi-honest adversaries, providing 128-bit security. 
We first provide implementation details and then show how to deal with floating-point numbers when conjunct the regression algorithms with cryptographic primitives.  
Finally, we present experimental results of \regsys\ on real-world datasets from the UCI ML repository. 
\vspace{-0.25cm} 
\subsection{Implementation and Dataset Details}
\par\noindent{\bf Implementation details.} 
We have implemented \regsys\ in C using GMP \cite{gmp} for large number operation, the FLINT library \cite{flint} for efficient secret sharing implementations and OpenSSL \cite{openssl} for implementing an authenticated channel.  
We choose the Joye-Libert (JL) cryptosystem \cite{JL-cryptosystem} to instantiate the additive homomorphic encryption scheme, over Paillier \cite{Paillier}, and Bresson et al. \cite{Bresson}  cryptosystems as its ciphertext size is 2$\times$ smaller.    
If the server is a cloud provider (e.g., Google, Microsoft, Apple), with the computation ability of the server, we can save the communication cost by a factor of two, which is quite reasonable for mobile applications. 
We implement the Joye-Libert cryptosystem using GMP. 
%
We implement an aggregation protocol that is a compilation of Bonawitz et al.'s \cite{SecureAggrePPML} and Mandal et al.'s \cite{secure-agg-cacr2018} protocols that handle user dropouts. 
The following crypto-primitives are used to implement the aggregation protocol under semi-honest adversaries. 
\begin{itemize}
 \item[--] We use the NIST P-256 curve from \textsf{OpenSSL} for the Elliptic-curve Diffie-Hellman (ECDH) pairwise key establishment. 
\item[--]  We use  \textsf{AES} in counter mode implemented using \textsf{AES-NI} instructions to implement \textsf{PRG}. 
\item[--] We use  \textsf{AES-GCM} from \textsf{OpenSSL} to implement a secure channel.  
\item[--] We use \textsf{SHA256} to  derive one-time pairwise keys. 
\item[--]  We implement Shamir's threshold secret sharing (\texttt{tss}) schemes, namely $(2,3)$-\texttt{tss} and  $(t,m)$-\texttt{tss} operations using the \textsf{FLINT} polynomial operation library, with NIST's 256-bit prime $q$.    
\end{itemize}
We have chosen the message space $\Z_{2^k}$ with $k = 256$ where $q < 2^k$, which is enough for handling high-dimension data used in our experiment.  
Our experiments were conducted on a desktop with a 3.40GHz Intel i7 and 12 GB RAM. The codes were compiled using \textsf{gcc 5.4.0} with \textsf{-std=c99 -O1 -fomit-frame-pointer} flag.

\par\noindent{\bf Datasets.}\label{sec:uci-dataset} 
We evaluate the performance of \regsys\ on 11 different real-world datasets from the UCI ML repository \cite{CancerDataset}. 
We use 6 different datasets with various sizes and dimensions, summarized in Table~\ref{tab:linear-reg-dataset}, for evaluating the linear/ridge regression training protocols, and 5 different datasets, summarized in Table~\ref{tab:logistic-reg-dataset}, for logistic regression training protocols. 
The dimension of the dataset ranges from 8 to 22 for linear regression and from 9 to 32 for logistic regression. 
For linear regression, we use the root mean squared error (RMSE) to measure the predictive accuracy of the model. 
For logistic regression, we measure the accuracy of the model using the standard method by computing a confusion matrix. 
Note that no privacy-preserving mechanism is used to compute the accuracy of the model. 

\begin{table}[h]
 \vspace{-0.15cm}
 \centering
 \caption{Dataset used in our experiment for linear/ridge regression models.}
 \label{tab:linear-reg-dataset}
 \vspace{-0.35cm}
 \begin{tabular}{c c c c c }\Xhline{2\arrayrulewidth}
 \textbf{Id} & \textbf{Name} & $n$ & $d$ & {\bf Reference} \\ \Xhline{2\arrayrulewidth}
 1 & Auto MPG  & 8 & 392 & \cite{AutoMPG} \\
 2 & Boston Housing Dataset & 14 & 506 & \cite{BostonHousing} \\
 3 & Energy Efficiency & 9 & 768 & \cite{EnergyEfficiency} \\
 4 & Wine Quality & 12 & 1,599 & \cite{WineDataset} \\ 
 5 & Parkinsons Telemonitoring & 22 & 5,875 & \cite{Telemonitoring} \\
 6 & Bike Sharing Dataset & 12 & 17,379 & \cite{BikeSharing} \\
  \Xhline{2\arrayrulewidth}
 \end{tabular}
 \vspace{-0.1cm}
\end{table}

\begin{table}[h]
 \vspace{-0.15cm}
 \centering
 \caption{Dataset used in our experiment for logistic regression model. $\dagger$ 20 out of 90 features were chosen.}
 \label{tab:logistic-reg-dataset}
  \vspace{-0.35cm}
 \begin{tabular}{c c c c c }\Xhline{2\arrayrulewidth}
 \textbf{Id} & \textbf{Name} & $n$ & $d$ & {\bf Reference} \\ \Xhline{2\arrayrulewidth}
 7 & Breast Cancer Dataset  & 32 & 454 & \cite{CancerDataset} \\
 8 & Credit Approval Dataset &  14 & 652 & \cite{CreditDataset} \\
 9  & Diabetes Dataset & 9 &  768 & \cite{DiabetesDataset} \\
 10 & Credit Card Clients & 24  & 30,000 & \cite{CreditCardClient} \\ 
 11 & US Census Income Dataset & $20^{\dagger}$ & 48,842  & \cite{CensusDataset,CensusIncomeDataset} \\
  \Xhline{2\arrayrulewidth}
 \end{tabular}
 \vspace{-0.1cm}
\end{table}

\begin{table}[h]
 \vspace{-0.15cm}
 \centering
 \caption{Number of the users and their training dataset size. Accuracy of the trained regression models achieved by $\regsys$.}
 \label{tab:DatasetPartition}
  \vspace{-0.35cm}
 \begin{tabular}{c c c c c c c c }\Xhline{2\arrayrulewidth}
  \multicolumn{4}{c}{Linear regression}  & \multicolumn{4}{c}{Logistic regression} \\ \hline
  ID & $m$ & $d_i$ & \textbf{RMSE} & ID & $m$ & $d_i$ & \textbf{Score} (\%) \\ \hline
  1 & 28 & 10 & 3.16 & 7 & 32 & 10 & 96.00  \\ 
  2 & 36 & 10 & 4.91 & 8 &  46 & 10 &  87.60 \\
  3 & 54 & 10 & 3.85 & 9 & 54 & 10 &  76.48 \\
  4 & 112 &  10 &   0.68 & 10 & 700 &  30 & 80.72   \\
  5 & 206 & 20 & 3.45 & 11 &  1140 & 30 &  89.60 \\
   6 & 609 & 20 & 147.80 &  &  & &   \\
  \hline
 \end{tabular}
 \vspace{-0.2cm}
\end{table}
\vspace{-0.4cm}
\subsection{Dealing with Floating Point Numbers}
In regression algorithms, the datasets and the model parameters are floating point numbers, both positive and negative, but the cryptographic techniques namely additive HE and secret sharing work over the finite ring of integers.
We provide the details about encoding floating point numbers to elements of $\Z_{2^{k}}$ and vice-versa via decoding. 

\par\noindent{\bf Encoding and decoding.} The encoding operation is applied on floating point numbers before performing cryptographic operations. 
As the message space is $\Z_{2^k}$, we divide the message space into two halves: the positive numbers are in $[0, 2^{k-1}-1]$, and the negative numbers are in $[2^{k-1}, 2^{k}-1]\equiv [-2^{k-1}, -1]$.
We convert each floating point number to an element of $\Z_{2^{k}}$ while maintaining its precision. 
Given an absolute floating point number $x$ in $\xbu^{(i)}$, the corresponding ring element, denoted as $\tilde{x}$ in $\Z_{2^{k}}$, is computed as $\tilde{x} = \texttt{FE}(x, \tau) = \texttt{round} (x \cdot 2^{\tau})$,
and each floating point number $y$ in $\{ y^{(i)}\}$ is converted to a ring element as $\tilde{y} = \texttt{FE}(y, 2\tau)$. If $x$ is negative, the corresponding ring element is $\tilde{x} = 2^k - \texttt{FE}(x, \tau)$, 
and similarly for $y$ in $\{ y^{(i)}\}$. 

Given $\tilde{x} \in \Z_{2^k}$, the decoding of $\tilde{x}$ is given by $x = \texttt{FD}(\tilde{x}, \tau) = -\frac{2^{k}-z}{2^{\tau}}$ if $\tilde{x} \geq 2^{k-1}$, otherwise $\frac{z}{2^{\tau}}$. 


\par\noindent{\bf Evaluating inner product.} 
Given $\thetabu \in \R^{n+1}$ and $\xbu \in \R^{n}$ and the corresponding vectors in $\Z_{2^k}$ are $\tilde{\thetabu} \in \Z_{2^k}^{n+1}$ and $\tilde{\xbu} \in \Z_{2^k}^{n}$ where 
$\tilde{\theta_0} = \texttt{FE}(\theta_0, 2\tau)$, $\tilde{\theta_i} = \texttt{FE}(\theta_i, \tau)$ and $\tilde{x_i} = \texttt{FE}(x_i, \tau)$. Then, $\tilde{\thetabu}\cdot\tilde{\xbu} = 2^{2\tau} \thetabu\cdot\xbu$. 
Thus, $\thetabu\cdot\xbu = \texttt{FD}(\tilde{\thetabu}\cdot\tilde{\xbu})$. If $\tilde{y} =  \texttt{FE}(-y, 2\tau)$, then $\thetabu\cdot\xbu - y = \texttt{FD}(\tilde{\thetabu}\cdot\tilde{\xbu} + \tilde{y})$. 

\par\noindent{\bf Evaluating sigmoid.}
We approximate the sigmoid function $\sigma(x)$ over $[-l, l]$ for some $l$ to a cubic polynomial as $\sigma_3(x) = c_0 + c_1 x + c_2 x^2 + c_3x^3$. 
Note that the coefficients of the polynomial are public. 
To evaluate $\sigma_3(z)$ over $\Z_{2^k}$, we convert the coefficients of $\sigma_3(x)$ as $q_0 = \texttt{FE}(c_0, 7\tau)$, $q_1 = \texttt{FE}(c_1, 5\tau)$, $q_2 = \texttt{FE}(c_2, 3\tau)$, and $q_3 = \texttt{FE}(c_3, \tau)$, i.e.,  $q_i = \texttt{FE}\big(c_i, (7-2*i)\tau\big)$ and 
$z \in \R$ as $\tilde{z} = \texttt{FE}(x, 2\tau)$. Then $\widetilde{\sigma_3(z)} = 2^{7\tau} \sigma_3(z)$, this implies $\sigma_3(z) = \texttt{FD}(\widetilde{\sigma_3(z)}, 7\tau)$, where $\tau$ is chosen so that there is no overflow in the message space $\Z_{2^k}$.

 \vspace{-0.25cm}
\subsection{Experimental Results}
This section presents the performance of \regsys\ where we report the timings, communication costs, and storage overhead for training and oblivious evaluation of linear and logistic regression algorithms for each user and the server.   
As the model privacy and data privacy while considering users dropping out in mobile applications have not been considered in previous work, e.g., \cite{SecureAggrePPML,secure-agg-cacr2018}, 
we do not compare the performance of \regsys\ with others in numerical values. However, we provide a system goalwise comparison in Section~\ref{sec:RelatedWork}.

\par\noindent {\bf Micro-benchmarking.} The additive homomorphic encryption and the aggregation protocol are two main operations that are frequently performed in the protocol. 
We perform micro-benchmarks that measure the timings of the basic operations, namely vector encryption, decryption and constant multiplication operations for the \textsf{JL} cryptosystem including floating point encoding and decoding operations, and the aggregation protocol for a user and the server to understand the deeper insight about the performance of the overall protocol. Table~\ref{tab:MicroBeanchmark} presents the timings for HE operations and the aggregation protocol. 
\begin{table}[h]
 \centering
  \vspace{-0.25cm}
 \caption{Time in milliseconds (ms) (using a single CPU) for vector encryption, decryption and const. multiplication of the \textsf{JL} cryptosystem and the aggregation protocol.}
 \label{tab:MicroBeanchmark}
  \vspace{-0.4cm}
  \resizebox{8.5cm}{!}{
 \begin{tabular}{ | l | l | l | l | l | l |}\cline{2-6}
  \multicolumn{1}{c}{}                 & \multicolumn{5}{|c|}{{\bf Vector dimension} ($n$)} \\ \hline
  {\bf Operations} & 10 & 20 & 30 & 40 & 50 \\ \hline
  Encryption ($\E(\xbu)$) & 8.39 & 16.05 & 24.05 & 32.35 & 39.83 \\ \hline
  Decryption ($\D(\xbu)$) & 4,027.02 & 8,131.76 & 12,134.78 & 16,178.63 & 20,107.34  \\ \hline
  Const. multiplication ($\E(\xbu)^{\zbu}$) & 60.56 & 128.78 & 181.83 & 248.58 &301.96  \\ \hline \cline{2-6}
   \multicolumn{1}{|c}{ {\bf Aggregation protocol} ($\pi_{\textsc{DeA}}$)}  & \multicolumn{5}{|c|}{{\bf Number of users ($m$)}} \\ \hline
 {\bf Dropout = 25\%}, $n = 30$ & 50 & 100 & 150 & 200 & 250 \\ \hline
 User time & 20.38 & 41.45 & 63.43 & 85.86 & 107.10 \\ \hline
 Server time & 34.79 & 161.36 & 420.90 & 823.35 & 1437.02 \\ \hline
 \end{tabular}}
  \vspace{-0.35cm}
\end{table}

\par\noindent {\bf Timing for training regression models.} For each dataset, we shuffle the dataset and  use approximately 70\% for training and 30\% for testing to calculate the accuracy of the model. The training dataset is then distributed into a set of $m$ users and each user holds an equal number of data points, given in Table~\ref{tab:DatasetPartition}. 
We randomly choose $2t$ users out of the $m$ users for their participations in computing the global gradient in the training phase where $t = \lceil \frac{m}{3} \rceil$ and $m$ is the total number of users. 
In our experiment, we set the number of dropout users in the gradient computation to $\delta = \lceil\frac{t}{2}\rceil$, and choose the aggregation size large enough to protect users' inputs privacy. For each dataset, the experiments were repeated five times, except the datasets Credit Card Clients and US Census Income datasets. 

We compare the accuracy of the trained model obtained using the privacy-preserving training protocol with the trained model obtained using the \texttt{sklearn} tool  (in clear, no security) for each dataset. Our achieved accuracy is very closed to the no security one ( which is close to the state-of-the-art accuracy). The number of iteration needed to achieve such accuracy is $R = 350$ iterations for linear/ridge regression and $R  = 300$ iterations for logistic regression, respectively. Figures~\ref{fig:LinTrainUserTime} and~\ref{fig:LogTrainUserTime} present the timings for training the linear and logistic regression models per user. Figure~\ref{fig:LinTrainServerTime} and~\ref{fig:LinTrainServerTime} show the server's timings for training the linear and logistic regression models.  All experiment results were obtained using a single CPU. In our experiment, no communication time is considered. 

For instance, to train the parkinsons telemonitoring data for a linear regression model, a user elapses about 105 milliseconds (ms) to compute the global gradient, and in the worst case, it elapses about 170 seconds (sec) for the entire training process. The ``worst case" is because of the fact that a user may be chosen a maximum of $R$ times.  On the other hand, the server's computation time is about 10 sec for each global gradient computation, and in total about 56.57 mins. 
To train the credit card clients data using a logistic regression model, a user elapses about 1066 ms to compute the global gradient, and in the worst case, a user elapses about 320 sec for the entire training process. The server's computation time is about 99 mins for each global gradient computation, and in total about  495.3 hours, which is because of performing a total of $4,219,200$ \textsf{JL} decryptions and 300 executions of $\pi_{\textsc{DeA}}$ with $2t$ users and $t/2$ dropouts where $t = 234$. {\it By exploiting parallelism using 24 CPUs, the training time for the server can be reduced to approximately 20 hours (estimated using the timings of unit operations).}  

For instance, according to our experimental settings, to train a linear model on the bike sharing dataset,  \regsys's computational overheads for the server and each user, compared to a normal system, are $2.7\times 10^5$ and $1.2 \times 10^6$, resp., where the normal system (naive implementation) provides no model and data privacy.
Similarly, for training a logistic regression model on the credit card clients dataset, \regsys's computational overheads for the server and each user are  $1.4\times 10^8$ and $3.2 \times 10^5$, resp.. The server's high computational overhead is due to the JL decryption algorithm.

Note that the communication latencies among users or the server interactions during the protocol execution are not included in our experiment. However, the total execution time of the training protocol will be the computation times plus the communication latencies.
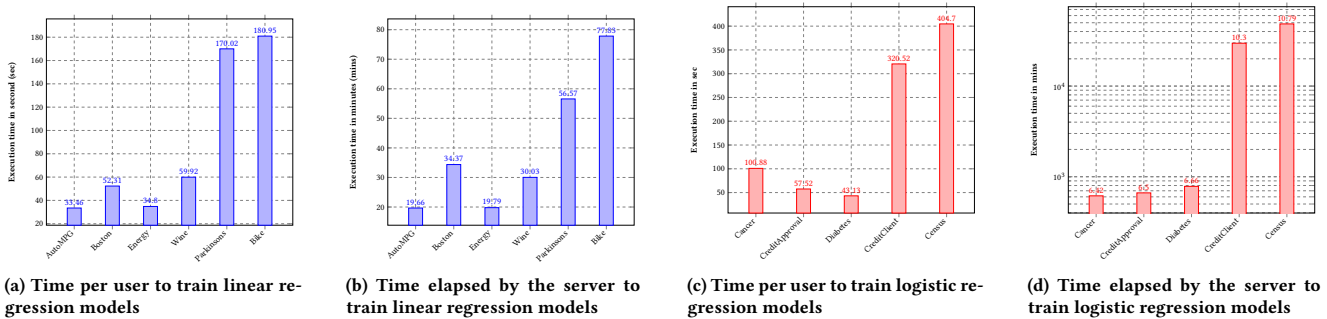
\begin{figure*}
 \vspace{-0.65cm}
    \centering
\subfloat[Time per user to train linear regression models]{\begin{tikzpicture}[scale=0.37]
\pgfplotsset{grid style={dashed,gray}}
\begin{axis}[
       ybar=-0.5cm,
        height=9cm, 
        bar width=0.5cm,
        ylabel={{\bf Execution time in second (sec)}},
      symbolic x coords={AutoMPG,Boston,Energy,Wine,Parkinsons,Bike},
        x tick label style={rotate=45, anchor=east, align=left},
        nodes near coords,
        nodes near coords align={vertical},
        grid=both, 
        enlarge x limits={abs=1cm}
        ]
 \addplot[blue,fill=blue!30!white] coordinates { (AutoMPG, 33.46)};
 \addplot[blue,fill=blue!30!white]  coordinates { (Boston, 52.31)};
  \addplot[blue,fill=blue!30!white]  coordinates { (Energy, 34.80 )};
   \addplot[blue,fill=blue!30!white]  coordinates { (Wine, 59.92)};
    \addplot[blue,fill=blue!30!white]  coordinates {(Parkinsons, 170.02) };
     \addplot[blue,fill=blue!30!white]  coordinates { (Bike, 180.95)};
\end{axis}
\end{tikzpicture}
\label{fig:LinTrainUserTime}
}%
   \qquad
\subfloat[Time elapsed by the  server to train linear regression models ]{\begin{tikzpicture}[thick, scale=0.37]
\pgfplotsset{grid style={dashed,gray}}
\begin{axis}[
       ybar=-0.5cm,
        height=9cm, 
        bar width=0.5cm,
        ylabel={{\bf Execution time in minutes (mins)}},
      symbolic x coords={AutoMPG,Boston,Energy,Wine,Parkinsons,Bike},
        x tick label style={rotate=45, anchor=east, align=left},
        nodes near coords,
        nodes near coords align={vertical},
        grid=both, 
        enlarge x limits={abs=1cm}
        ]
 \addplot[blue,fill=blue!30!white] coordinates { (AutoMPG, 19.66)};
 \addplot[blue,fill=blue!30!white]  coordinates { (Boston, 34.37)};
  \addplot[blue,fill=blue!30!white]  coordinates { (Energy, 19.785)};
   \addplot[blue,fill=blue!30!white]  coordinates { (Wine, 30.03)};
    \addplot[blue,fill=blue!30!white]  coordinates {(Parkinsons, 56.57) };
     \addplot[blue,fill=blue!30!white]  coordinates { (Bike, 77.83)};
\end{axis}
\end{tikzpicture}
\label{fig:LinTrainServerTime}
}%
 \qquad
  \subfloat[Time per user to train logistic regression models]{\begin{tikzpicture}[scale=0.37]
  \pgfplotsset{grid style={dashed,gray}}
 \begin{axis}[
       ybar=-0.5cm,
        height=9cm, 
        bar width=0.5cm,
        ylabel={{\bf Execution time in sec}},
        symbolic x coords={Cancer,CreditApproval,Diabetes,CreditClient,Census},
        x tick label style={rotate=45, anchor=east, align=left},
        nodes near coords,
        nodes near coords align={vertical},
       grid=both, 
        enlarge x limits={abs=1cm},
         xtick={Cancer,CreditApproval,Diabetes,CreditClient,Census},
        ]
   \addplot[red,fill=red!30!white] coordinates { (Cancer, 100.88)};
       \addplot[red,fill=red!30!white] coordinates { (CreditApproval,  57.516)};
         \addplot [red,fill=red!30!white]coordinates { (Diabetes,  43.133) };
           \addplot[red,fill=red!30!white] coordinates { (CreditClient, 320.52) };
             \addplot[red,fill=red!30!white] coordinates {(Census,  404.70) };
 \end{axis}
  \end{tikzpicture}
  \label{fig:LogTrainUserTime}
  }%
     \qquad
  \subfloat[Time elapsed by the  server to train logistic regression models]{\begin{tikzpicture}[thick, scale=0.37]
  \pgfplotsset{grid style={dashed,gray}}
  \begin{axis}[
       ybar=-0.5cm,
        height=9cm, 
        bar width=0.5cm,
        ylabel={{\bf Execution time in mins}},
        symbolic x coords={Cancer,CreditApproval,Diabetes,CreditClient,Census},
        x tick label style={rotate=45, anchor=east, align=left},
        nodes near coords,
        nodes near coords align={vertical},
        grid=both, 
        enlarge x limits={abs=1cm},
       ymode=log,
        xtick={Cancer,CreditApproval,Diabetes,CreditClient,Census},
        ]
       \addplot[red,fill=red!30!white] coordinates { (Cancer, 615.56)};
       \addplot[red,fill=red!30!white] coordinates { (CreditApproval, 663.71 )};
         \addplot [red,fill=red!30!white]coordinates { (Diabetes, 783.18 ) };
           \addplot[red,fill=red!30!white] coordinates { (CreditClient, 29719.37) };
             \addplot[red,fill=red!30!white] coordinates {(Census, 48519.24 ) };
 \end{axis}
  \end{tikzpicture}
  \label{fig:LogTrainServerTime}
 }%
 \vspace{-0.35cm}
 \caption{Time in seconds for a user and in minutes for the server to train linear and logistic regression models.} 
  \label{fig:MobileServerTiming}%
\vspace{-0.25cm}
\end{figure*}

\vspace{-0.3cm} 
\par\noindent {\bf Timing for oblivious regression predictions.}
We also measure the performance of the oblivious evaluations of linear and regression models on the aforementioned datasets. The performance of the oblivious predictions depends on the dimension of the dataset. 
Table~\ref{tab:ObliviousPredTime} summarizes the timings of a user and the server for obliviously evaluating linear and logistic regression models. The experiments are repeated ten times for each point. From Table~\ref{tab:ObliviousPredTime}, we can observe that the time for the user is larger than that of the server, which is due to the cost of the JL decryption operation. 
\begin{table}[h]
 \vspace{-0.15cm}
 \centering
 \caption{Time in milliseconds for oblivious regression prediction for the server and the user.} 
 \label{tab:ObliviousPredTime}
  \vspace{-0.3cm}
 \begin{tabular}{c c  c c c  c }\Xhline{2\arrayrulewidth}
  \multicolumn{3}{c}{Linear regression}  & \multicolumn{3}{c}{Logistic regression} \\ \hline
  ID &  \textbf{User} & \textbf{Server} & ID & \textbf{User}  & \textbf{Server} \\ \hline
  1 &  430.15 & 4.12 & 7 & 935.04 &  11.07  \\ 
  2 & 437.35 & 6.06 & 8 &  891.69 &  10.51  \\
  3 & 436.55  & 3.04 & 9 & 883.78 &  6.52  \\
  4 &  443.94 &  4.32  & 10 & 893.12  & 15.85  \\
  5 &  446.65 & 8.21 & 11 &  950.56 & 36.95  \\
   6 & 422.17  & 3.22 &  &   &   \\
  \hline
 \end{tabular}
 \vspace{-0.2cm}
\end{table}
\par\noindent {\bf Communication and storage cost.} 
For training a linear regression model, the communication cost in the training phase involves receiving the encrypted model, transmitting the encrypted share of the local gradient and information exchange in the share reconstruction phase. For training a logistic regression model, in addition to the linear regression's communication cost, one more round of information exchange is required in the shared local gradient computation phase. 
Figure~\ref{fig:TrainingBitsTransmission} presents the maximum amount of bits a user needs to transfer to accomplish the training phase for the choice of the parameters and the achieved accuracy in Table~\ref{tab:DatasetPartition}.  
Note that the communication cost for a user depends on how many times the user is chosen in the training phase. 
  For instance, for training the parkinsons telemonitoring data, in the worst case, a user need to exchange 12.05 Megabyte (MB) of data to train the linear regression model. 
On the other hand, for the credit card clients data, in the worst case, a user need to exchange 36.10 Megabyte (MB) of data to train the logistic regression model. 
 The communication cost for the server for each dataset is provided in Figure~\ref{fig:ServerTrainingBitsTransmission}. 
The storage overheads for each user and the server is presented in Figures~\ref{fig:UserStorage} and~\ref{fig:ServerStorage}, respectively.  

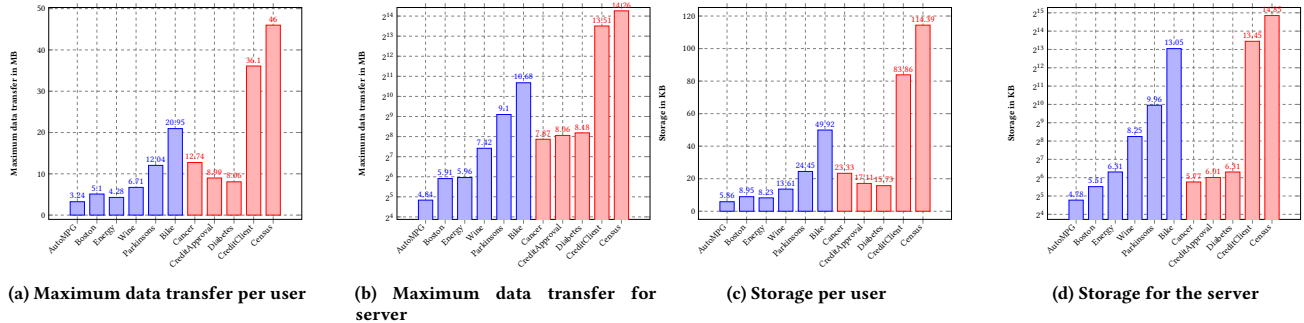
\begin{figure*}
 \vspace{-0.55cm}
    \centering
\subfloat[Maximum data transfer per user]{\begin{tikzpicture}[scale=0.38]
\pgfplotsset{grid style={dashed,gray}}
\begin{axis}[
       ybar=-0.5cm,
        height=9cm, 
        bar width=0.5cm,
        ylabel={{\bf Maximum data transfer in MB}},
        symbolic x coords={AutoMPG,Boston,Energy,Wine,Parkinsons,Bike,Cancer,CreditApproval, Diabetes,CreditClient,Census},
        x tick label style={rotate=45, anchor=east, align=left},
        nodes near coords,
        nodes near coords align={vertical},
        grid=both, 
        enlarge x limits={abs=1cm},
        xtick={AutoMPG,Boston,Energy,Wine,Parkinsons,Bike,Cancer,CreditApproval, Diabetes,CreditClient,Census},
        legend pos=north west,
        ]
      \addplot[blue,fill=blue!30!white] coordinates { (AutoMPG, 3.240509)};
 \addplot[blue,fill=blue!30!white]  coordinates { (Boston, 5.100647)};
  \addplot[blue,fill=blue!30!white]  coordinates { (Energy, 4.279510)};
   \addplot[blue,fill=blue!30!white]  coordinates { (Wine, 6.710663)};
    \addplot[blue,fill=blue!30!white]  coordinates {(Parkinsons, 12.041107) };
     \addplot[blue,fill=blue!30!white]  coordinates { (Bike, 20.948334)};
     
      \addplot[red,fill=red!30!white] coordinates { (Cancer, 12.742432)};
       \addplot[red,fill=red!30!white] coordinates { (CreditApproval,  8.988098)};
         \addplot [red,fill=red!30!white]coordinates { (Diabetes,  8.063690) };
           \addplot[red,fill=red!30!white] coordinates { (CreditClient, 36.103851) };
             \addplot[red,fill=red!30!white] coordinates {(Census,  46.00) };
\end{axis}
\end{tikzpicture}
\label{fig:TrainingBitsTransmission}
}%
   \qquad
  \subfloat[Maximum data transfer for server]{\begin{tikzpicture}[scale=0.38]
  \pgfplotsset{grid style={dashed,gray}}
\begin{axis}[
       ybar=-0.5cm,
        height=9cm, 
        bar width=0.5cm,
        ylabel={{\bf Maximum data transfer in MB}},
      symbolic x coords={AutoMPG,Boston,Energy,Wine,Parkinsons,Bike,Cancer,CreditApproval, Diabetes,CreditClient,Census},
      ymin=0, ymax=22000,
        x tick label style={rotate=45, anchor=east, align=left},
        nodes near coords,
        nodes near coords align={vertical},
        grid=both, 
        enlarge x limits={abs=1cm},
        xtick={AutoMPG,Boston,Energy,Wine,Parkinsons,Bike,Cancer,CreditApproval, Diabetes,CreditClient,Census},
        ymode=log,
        log basis y={2},
        legend pos=north west,
        ]
      \addplot[blue,fill=blue!30!white] coordinates { (AutoMPG, 28.537048)};
 \addplot[blue,fill=blue!30!white]  coordinates { (Boston, 60.045837)};
  \addplot[blue,fill=blue!30!white]  coordinates { (Energy,  62.268951)};
   \addplot[blue,fill=blue!30!white]  coordinates { (Wine, 170.871094)};
    \addplot[blue,fill=blue!30!white]  coordinates {(Parkinsons, 548.057526 ) };
     \addplot[blue,fill=blue!30!white]  coordinates { (Bike, 1638.085236)};
      \addplot[red,fill=red!30!white] coordinates { (Cancer,  233.196716)};
       \addplot[red,fill=red!30!white] coordinates { (CreditApproval,  265.987610 )};
         \addplot [red,fill=red!30!white]coordinates { (Diabetes, 290.693024  ) };
           \addplot[red,fill=red!30!white] coordinates { (CreditClient,  11631.19751) };
             \addplot[red,fill=red!30!white] coordinates {(Census, 19602.246124  ) };
\end{axis}
  \end{tikzpicture}
  \label{fig:ServerTrainingBitsTransmission}
  }%
  \subfloat[Storage per user]{\begin{tikzpicture}[scale=0.38]
\pgfplotsset{grid style={dashed,gray}}
\begin{axis}[
       ybar=-0.5cm,
        height=9cm, 
        bar width=0.5cm,
        ylabel={{\bf Storage in KB}},
      symbolic x coords={AutoMPG,Boston,Energy,Wine,Parkinsons,Bike,Cancer,CreditApproval, Diabetes,CreditClient,Census},
       x tick label style={rotate=45, anchor=east, align=left},
        nodes near coords,
        nodes near coords align={vertical},
        grid=both, 
        enlarge x limits={abs=1cm},
        xtick={AutoMPG,Boston,Energy,Wine,Parkinsons,Bike,Cancer,CreditApproval, Diabetes,CreditClient,Census},
        legend pos=north west,
        ]
      \addplot[blue,fill=blue!30!white] coordinates { (AutoMPG, 5.859375 )};
 \addplot[blue,fill=blue!30!white]  coordinates { (Boston, 8.953125)};
  \addplot[blue,fill=blue!30!white]  coordinates { (Energy, 8.234375)};
   \addplot[blue,fill=blue!30!white]  coordinates { (Wine, 13.609375 )};
    \addplot[blue,fill=blue!30!white]  coordinates {(Parkinsons,  24.453125) };
     \addplot[blue,fill=blue!30!white]  coordinates { (Bike, 49.921875 )};
      \addplot[red,fill=red!30!white] coordinates { (Cancer, 23.328125)};
       \addplot[red,fill=red!30!white] coordinates { (CreditApproval,  17.109375)};
         \addplot [red,fill=red!30!white]coordinates { (Diabetes,  15.734375) };
           \addplot[red,fill=red!30!white] coordinates { (CreditClient, 83.859375) };
             \addplot[red,fill=red!30!white] coordinates {(Census,   114.390625) };
             
\end{axis}
\end{tikzpicture}
\label{fig:UserStorage}
}%
   \qquad
  \subfloat[Storage for the server]{\begin{tikzpicture}[scale=0.38]
  \pgfplotsset{grid style={dashed,gray}}
\begin{axis}[
       ybar=-0.5cm,
        height=9cm, 
        bar width=0.5cm,
        ylabel={{\bf Storage in KB}},
      symbolic x coords={AutoMPG,Boston,Energy,Wine,Parkinsons,Bike,Cancer,CreditApproval, Diabetes,CreditClient,Census},
      ymin=0, ymax=40000,
        x tick label style={rotate=45, anchor=east, align=left},
        nodes near coords,
        nodes near coords align={vertical},
        grid=both, 
        enlarge x limits={abs=1cm},
        ymode=log,
        log basis y={2},
        xtick={AutoMPG,Boston,Energy,Wine,Parkinsons,Bike,Cancer,CreditApproval, Diabetes,CreditClient,Census},
        legend pos=north west,
        ]
      \addplot[blue,fill=blue!30!white] coordinates { (AutoMPG, 27.375000)};
 \addplot[blue,fill=blue!30!white]  coordinates { (Boston, 45.593750 )};
  \addplot[blue,fill=blue!30!white]  coordinates { (Energy,  79.406250)};
   \addplot[blue,fill=blue!30!white]  coordinates { (Wine, 303.562500)};
    \addplot[blue,fill=blue!30!white]  coordinates {(Parkinsons, 992.375000 ) };
     \addplot[blue,fill=blue!30!white]  coordinates { (Bike,  8441.812500 )};
      \addplot[red,fill=red!30!white] coordinates { (Cancer, 54.408691  )};
       \addplot[red,fill=red!30!white] coordinates { (CreditApproval,  64.408691)};
         \addplot [red,fill=red!30!white]coordinates { (Diabetes,  79.408691) };
           \addplot[red,fill=red!30!white] coordinates { (CreditClient, 11140.382324) };
             \addplot[red,fill=red!30!white] coordinates {(Census,  29525.726074 ) };
\end{axis}
  \end{tikzpicture}
  \label{fig:ServerStorage}
  }%
  \vspace{-0.45cm}
  \caption{Total data exchange in Megabyte (MB) and storage required in Kilobyte (KB)  per user and the server to train linear and logistic regression models.} 
  \label{fig:UserServerTrainingStorageBitsTransmission}%
\vspace{-0.53cm}
\end{figure*}

\par\noindent{\bf Optimizations.} 
It can be observed from the protocol that there are many scopes of paralleling the server's computation (e.g., JL vector decryption, encrypted vector multiplication, secret reconstruction and PRG evaluation, etc). Exploiting such parallelism using multiple CPUs, the server's computational time will be reduced to several magnitudes. While generating the random vector $\rbu$, users can generate it using a PRG. 

\par\noindent{\bf Discussions.} In \regsys, one can take a garbled circuit or purely secret-sharing based approach to implement the shared local gradient protocol. 
As the local gradient computation for linear regression can be expressed as a circuit of multiplicative depth two, a secret-sharing based approach will require at least two rounds of communications plus an additional cost of generating multiplication triplets in an offline phase. 
If the multiplication of two $b$-bit floating point numbers needs approximately $b^{1.58}$ gates (according to the Karatsuba algorithm), 
the computation of the local gradient between the server and a user has a transmission cost of at least $8\times 128\times b^{1.58}\times n = 2^{10} \times n\times b^{1.58}$ plus the transmission cost of the OT protocol using a garbled circuit approach for an $n$-dimensional vector. 
For an accuracy of six decimal places, this cost is much larger than $2n\log_2(N) $ bits for transmitting the encrypted model as well as the local gradient using the \textbf{JL} cryptosystem where $N$ is a modulus of size 3072 bits. 
Similarly, for the logistic regression, the transmission cost is roughly $2^9 \times (2n+1) b^{1.58}$ using the garbled circuit approach, which is larger than $3n\log_2(N)$ for each data point.  
As mobile applications may have low bandwidth and slow connections, a garbled circuit or purely secret-sharing based approach will incur a higher communication overhead compared to an additive HE-based approach.
\vspace{-0.25cm}

\section{Related Work}\label{sec:RelatedWork}
\par\noindent{\bf Privacy-preserving Linear Regression.}	
Privacy-preserving computation of linear regression has received considerable attention. 
Early works \cite{priv-lin-reg-siam2004,priv-reg-kdd2004,regression-mpc2004,Karr-linreg-2009,karr2005secure-lin-reg} have considered learning linear regression model on either horizontally or vertically distributed datasets. 
Hall et al. \cite{hall-jos2011} proposed  protocols for linear regression based on homomorphic encryption techniques. 
Nikolaenko et al. \cite{Nikolaenko:2013SP} proposed a system for privacy-preserving computation of the ridge regression model by combining homomorphic encryption and Yao garbled circuits. 
In a follow-up work, Gasc{\'o}n et al. \cite{LinReg:PETS2017} proposed protocols for linear regression models based on hybrid-MPC techniques and Yao garbled circuits and using the conjugate gradient descent algorithm. 
Bogdanov et al. \cite{rmind-ieee-dsc-2018} developed tools for privacy-preserving linear regression based on secret sharing.  
Other approach for privacy-preserving linear regression is based on fully homomorphic encryption (FHE) scheme \cite{fhe}, which may not be suitable for mobile applications. 
The work of \cite{MLConfidential} can be applied to linear regression for the settings of MLaaS.  

\par\noindent{\bf Privacy-preserving Logistic Regression.} 
Logistic regression is  an essential technique to  classify  data. Aono et al. \cite{LogisReg} proposed a system for both training and predicting data using logistic regression relying on  additive  homomorphic  encryption where 
they considered the computation outsourcing scenario in which a server computes the logistic regression model and sends an encrypted model to the user. 
They also showed how to make the system differential privacy enabled. 
Bonte and Vercauteren \cite{LogisReg-HE} explored secure training for logistic regression using somewhat homomorphic encryption in the computation outsourcing scenario where the training is done using Newton-Raphson method. 
Kim et al. \cite{LogisticReg-HE-MedInform} also investigated the training phase of the logistic regression model, using somewhat homomorphic encryption scheme, based on gradient descent algorithm and the  Taylor series polynomial approximation.  
Zhu et al. \cite{LogisReg-Cloud} presented a secure outsourcing protocol for training and evaluating logistic regression classifier in cloud.  
Kim et al. in \cite{Kim2018} proposed a method to train a logistic regression model based on the approximate homomorphic encryption. There is no satisfactory solution for training an ML model using SWHE or FHE when data come from multiple sources. All these techniques cannot handle the dropout scenario and are not suitable for mobile applications.

\par\noindent{\bf Privacy-preserving Federated Learning.} 
Shokri and Shmatikov \cite{PPDL} presented a scheme for privacy-preserving deep learning. 
Hardy et al. \cite{PriFederated} presented a three-party protocol for logistic regression using additive homomorphic encryption 
where the protocol consists of privacy-preserving entity resolution and federated logistic regression and the data is vertically partitioned. Their protocol cannot handle the dropout scenario, and 
privacy is not ensured when the dataset contains only one point. Note that their work have not considered the model privacy.  Our solution is more general compared to theirs. 
Fioretto and Hentenryck in \cite{FedDataSharing} proposed a protocol for federated data sharing to use under the framework of differential privacy, which allows the users to release a privacy-preserving version of the dataset. 
This dataset can be used to train various predictors for linear regression, logistic regression, and support vector machines.
Liu et al. \cite{FedTransLearning} proposed a technique for privacy-preserving federated transfer learning. 
Truex et al. \cite{FedLearning-HE-DP2018} presented an approach for private federated learning, decision tree, neural networks that provides data  privacy  guarantees  using  differential  privacy  and threshold homomorphic encryption schemes.

\par\noindent{\bf Generic Secure ML Systems.}
Systems, namely \textsf{SecureML} \cite{SecureML:SP}, \textsf{Prio} \cite{Prio-Usenix2017} and \textsf{ABY}$^3$ \cite{ABY3} can be used to perform privacy-preserving linear and logistic regression, 
but such systems need existence of addition servers (other than the server used for coordination) where users secret share their data among a set of servers. 
These systems provide stronger security compared to ours because in \regsys, the server has access to the model in each iteration of the model update, which is quite suitable in the federated setting.  
However, these approaches are orthogonal to the federated learning setting where users send (using secret sharing) their private data to the servers, but in federated learning, users' data never go out of the devices. 
\vspace{-0.3cm}
\section{Concluding Remarks and Future Work}
This paper presented \regsys, a privacy-preserving system for training and oblivious predictions of the predictive models such as linear and logistic regressions in the federated setting, while ensuring dropout robustness and the data as well as the model privacy.  
\regsys\ enables a robust and secure training process by iteratively executing a secure multiparty global gradient protocol built using lightweight cryptographic primitives suitable for mobile applications.  
The security of  \regsys\ is analyzed against semi-honest adversaries. Our experimental results on several real-world datasets demonstrate the practicality of \regsys\ to incorporate in the federated learning system.  

We intend to implement \regsys\ on smartphones to evaluate its efficiency. As today's server platforms are equipped with multiple CPUs, we are working on designing an easy-and-parallel implementation interface at the server-side to reduce the computation time.  While keeping the general protocol flow in Section~\ref{sec:ProtocolFlow} same, we are extending our work to neural networks in a separate paper. As a future work, it would be interesting to investigate the case of users asynchronously participating in the training phase. 
\vspace{0.1cm}
\par\noindent{\bf Acknowledgement.} 
This work is supported by the NSERC Discovery grant. 
The authors would like to thank the anonymous reviewers of CCSW2019 for their insightful comments and suggestions to improve the quality of the paper.  

 \bibliographystyle{acm}
 \bibliography{ref-reg}

\begin{thebibliography}{10}

\bibitem{AItype}
Ai.type.
\newblock \url{https://www.androidauthority.com/ai-type-data-exposed-820431/}.

\bibitem{LogisReg-HE}
{\sc Aono, Y., Hayashi, T., Trieu~Phong, L., and Wang, L.}
\newblock Scalable and secure logistic regression via homomorphic encryption.
\newblock In {\em Proceedings of the Sixth ACM Conference on Data and
  Application Security and Privacy\/} (2016), ACM, pp.~142--144.

\bibitem{AutoMPG}
{\sc {Auto MPG data set.}}
\newblock \url{https://archive.ics.uci.edu/ml/datasets/auto+mpg}, 1993.
\newblock Online; accessed 29 July 2019.

\bibitem{LabeledHE}
{\sc Barbosa, M., Catalano, D., and Fiore, D.}
\newblock Labeled homomorphic encryption.
\newblock In {\em Computer Security -- ESORICS 2017\/} (Cham, 2017), S.~N.
  Foley, D.~Gollmann, and E.~Snekkenes, Eds., Springer International
  Publishing, pp.~146--166.

\bibitem{rmind-ieee-dsc-2018}
{\sc {Bogdanov}, D., {Kamm}, L., {Laur}, S., and {Sokk}, V.}
\newblock Rmind: A tool for cryptographically secure statistical analysis.
\newblock {\em IEEE Transactions on Dependable and Secure Computing 15}, 3 (May
  2018), 481--495.

\bibitem{TowardsFL}
{\sc Bonawitz, K., Eichner, H., Grieskamp, W., Huba, D., Ingerman, A., Ivanov,
  V., Kiddon, C., Konecn{\'{y}}, J., Mazzocchi, S., McMahan, H.~B., Overveldt,
  T.~V., Petrou, D., Ramage, D., and Roselander, J.}
\newblock Towards federated learning at scale: System design.
\newblock {\em CoRR abs/1902.01046\/} (2019).

\bibitem{SecureAggrePPML}
{\sc Bonawitz, K., Ivanov, V., Kreuter, B., Marcedone, A., McMahan, H.~B.,
  Patel, S., Ramage, D., Segal, A., and Seth, K.}
\newblock Practical secure aggregation for privacy-preserving machine learning.
\newblock In {\em Proceedings of the 2017 ACM SIGSAC Conference on Computer and
  Communications Security\/} (New York, NY, USA, 2017), CCS '17, ACM,
  pp.~1175--1191.

\bibitem{LogisReg}
{\sc Bonte, C., and Vercauteren, F.}
\newblock Privacy-preserving logistic regression training.
\newblock Tech. rep., IACR Cryptology ePrint Archive 233, 2018.

\bibitem{BostonHousing}
{\sc {Boston Housing Dataset.}}
\newblock
  \url{https://archive.ics.uci.edu/ml/machine-learning-databases/housing/},
  2019.
\newblock Online; accessed 29 July 2019.

\bibitem{Bresson}
{\sc Bresson, E., Catalano, D., and Pointcheval, D.}
\newblock A simple public-key cryptosystem with a double trapdoor decryption
  mechanism and its applications.
\newblock In {\em Advances in Cryptology - ASIACRYPT 2003\/} (Berlin,
  Heidelberg, 2003), C.-S. Laih, Ed., Springer Berlin Heidelberg, pp.~37--54.

\bibitem{TwoIsNotEnough}
{\sc Buescher, N., Boukoros, S., Bauregger, S., and Katzenbeisser, S.}
\newblock Two is not enough: Privacy assessment of aggregation schemes in smart
  metering.
\newblock {\em Proceedings on Privacy Enhancing Technologies 2017}, 4 (2017),
  198--214.

\bibitem{CensusIncomeDataset}
{\sc {Census Income Data Set.}}
\newblock \url{https://archive.ics.uci.edu/ml/datasets/census+income}, 1996.
\newblock Online; accessed 29 July 2019.

\bibitem{Prio-Usenix2017}
{\sc Corrigan-Gibbs, H., and Boneh, D.}
\newblock Prio: Private, robust, and scalable computation of aggregate
  statistics.
\newblock In {\em Proceedings of the 14th USENIX Conference on Networked
  Systems Design and Implementation\/} (Berkeley, CA, USA, 2017), NSDI'17,
  USENIX Association, pp.~259--282.

\bibitem{WineDataset}
{\sc Cortez, P., Cerdeira, A., Almeida, F., Matos, T., and Reis, J.}
\newblock Modeling wine preferences by data mining from physicochemical
  properties.
\newblock {\em Decision Support Systems 47}, 4 (2009), 547 -- 553.
\newblock Smart Business Networks: Concepts and Empirical Evidence.

\bibitem{DiabetesDataset}
{\sc {Diabetes Data Set}}.
\newblock \url{https://archive.ics.uci.edu/ml/datasets/diabetes}, 1994.

\bibitem{DHKeyAgreement}
{\sc Diffie, W., and Hellman, M.}
\newblock New directions in cryptography.
\newblock {\em IEEE Trans. Inf. Theor. 22}, 6 (Sept. 2006), 644--654.

\bibitem{priv-lin-reg-siam2004}
{\sc Du, W., Han, Y.~S., and Chen, S.}
\newblock {\em Privacy-Preserving Multivariate Statistical Analysis: Linear
  Regression and Classification}.
\newblock pp.~222--233.

\bibitem{CreditDataset}
{\sc Dua, D., and Graff, C.}
\newblock {UCI} machine learning repository.
\newblock \url{https://archive.ics.uci.edu/ml/datasets/credit+approval}, 2017.

\bibitem{CancerDataset}
{\sc Dua, D., and Graff, C.}
\newblock {UCI} machine learning repository.
\newblock
  \url{https://archive.ics.uci.edu/ml/datasets/Breast+Cancer+Wisconsin+(Diagnostic)},
  2017.

\bibitem{diff-privacy}
{\sc Dwork, C., and Roth, A.}
\newblock The algorithmic foundations of differential privacy.
\newblock {\em Found. Trends Theor. Comput. Sci. 9}, 3\&\#8211;4 (Aug. 2014),
  211--407.

\bibitem{BikeSharing}
{\sc Fanaee-T, H., and Gama, J.}
\newblock Event labeling combining ensemble detectors and background knowledge.
\newblock {\em Progress in Artificial Intelligence\/} (2013), 1--15.

\bibitem{FedDataSharing}
{\sc Fioretto, F., and Van~Hentenryck, P.}
\newblock Privacy-preserving federated data sharing.
\newblock In {\em Proceedings of the 18th International Conference on
  Autonomous Agents and MultiAgent Systems\/} (Richland, SC, 2019), AAMAS '19,
  International Foundation for Autonomous Agents and Multiagent Systems,
  pp.~638--646.

\bibitem{ModelInversionAttacks}
{\sc Fredrikson, M., Jha, S., and Ristenpart, T.}
\newblock Model inversion attacks that exploit confidence information and basic
  countermeasures.
\newblock In {\em Proceedings of the 22Nd ACM SIGSAC Conference on Computer and
  Communications Security\/} (New York, NY, USA, 2015), CCS '15, ACM,
  pp.~1322--1333.

\bibitem{gasconpets2017}
{\sc Gasc{\'o}n, A., Schoppmann, P., Balle, B., Raykova, M., Doerner, J.,
  Zahur, S., and Evans, D.}
\newblock Privacy-preserving distributed linear regression on high-dimensional
  data.
\newblock {\em Proceedings on Privacy Enhancing Technologies 2017}, 4 (2017),
  345--364.

\bibitem{LinReg:PETS2017}
{\sc Gasc{\'o}n, A., Schoppmann, P., Balle, B., Raykova, M., Doerner, J.,
  Zahur, S., and Evans, D.}
\newblock Privacy-preserving distributed linear regression on high-dimensional
  data.
\newblock {\em Proceedings on Privacy Enhancing Technologies 2017}, 4 (2017),
  345--364.

\bibitem{fhe}
{\sc Gentry, C.}
\newblock Fully homomorphic encryption using ideal lattices.
\newblock In {\em Proceedings of the Forty-first Annual ACM Symposium on Theory
  of Computing\/} (New York, NY, USA, 2009), STOC '09, ACM, pp.~169--178.

\bibitem{SecureScalarProduct}
{\sc Goethals, B., Laur, S., Lipmaa, H., and Mielik{\"a}inen, T.}
\newblock On private scalar product computation for privacy-preserving data
  mining.
\newblock In {\em Information Security and Cryptology -- ICISC 2004\/} (Berlin,
  Heidelberg, 2005), C.-s. Park and S.~Chee, Eds., Springer Berlin Heidelberg,
  pp.~104--120.

\bibitem{MLConfidential}
{\sc Graepel, T., Lauter, K., and Naehrig, M.}
\newblock Ml confidential: Machine learning on encrypted data.
\newblock In {\em Information Security and Cryptology -- ICISC 2012\/} (Berlin,
  Heidelberg, 2013), T.~Kwon, M.-K. Lee, and D.~Kwon, Eds., Springer Berlin
  Heidelberg, pp.~1--21.

\bibitem{gmp}
{\sc Granlund, T., et~al.}
\newblock {\sc GMP}: the \textsc{GNU} multiple precision arithmetic library,
  1991.

\bibitem{hall-jos2011}
{\sc Hall, R., Fienberg, S.~E., and Nardi, Y.}
\newblock Secure multiple linear regression based on homomorphic encryption.
\newblock {\em Journal of Official Statistics 27}, 4 (2011), 669.

\bibitem{PriFederated}
{\sc Hardy, S., Henecka, W., Ivey{-}Law, H., Nock, R., Patrini, G., Smith, G.,
  and Thorne, B.}
\newblock Private federated learning on vertically partitioned data via entity
  resolution and additively homomorphic encryption.
\newblock {\em CoRR abs/1711.10677\/} (2017).

\bibitem{flint}
{\sc Hart, W., Johansson, F., and Pancratz, S.}
\newblock {FLINT}: {F}ast {L}ibrary for {N}umber {T}heory, 2013.
\newblock Version 2.4.0, \url{http://flintlib.org}.

\bibitem{JL-cryptosystem}
{\sc Joye, M., and Libert, B.}
\newblock Efficient cryptosystems from 2k-th power residue symbols.
\newblock In {\em Advances in Cryptology -- EUROCRYPT 2013\/} (Berlin,
  Heidelberg, 2013), T.~Johansson and P.~Q. Nguyen, Eds., Springer Berlin
  Heidelberg, pp.~76--92.

\bibitem{regression-mpc2004}
{\sc Karr, A.~F., Lin, X., Sanil, A.~P., and Reiter, J.~P.}
\newblock Regression on distributed databases via secure multi-party
  computation.
\newblock In {\em Proceedings of the 2004 Annual National Conference on Digital
  Government Research\/} (2004), dg.o '04, Digital Government Society of North
  America, pp.~108:1--108:2.

\bibitem{karr2005secure-lin-reg}
{\sc Karr, A.~F., Lin, X., Sanil, A.~P., and Reiter, J.~P.}
\newblock Secure regression on distributed databases.
\newblock {\em Journal of Computational and Graphical Statistics 14}, 2 (2005),
  263--279.

\bibitem{Karr-linreg-2009}
{\sc Karr, A.~F., Lin, X., Sanil, A.~P., and Reiter, J.~P.}
\newblock Privacy-preserving analysis of vertically partitioned data using
  secure matrix products.
\newblock {\em J. Official Statistics\/} (2009).

\bibitem{SecureStat}
{\sc Kiltz, E., Leander, G., and Malone-Lee, J.}
\newblock Secure computation of the mean and related statistics.
\newblock In {\em Theory of Cryptography\/} (Berlin, Heidelberg, 2005),
  J.~Kilian, Ed., Springer Berlin Heidelberg, pp.~283--302.

\bibitem{Kim2018}
{\sc Kim, A., Song, Y., Kim, M., Lee, K., and Cheon, J.~H.}
\newblock Logistic regression model training based on the approximate
  homomorphic encryption.
\newblock {\em BMC Medical Genomics 11}, 4 (Oct 2018), 83.

\bibitem{LogisticReg-HE-MedInform}
{\sc Kim, M., Song, Y., Wang, S., Xia, Y., and Jiang, X.}
\newblock Secure logistic regression based on homomorphic encryption: Design
  and evaluation.
\newblock {\em JMIR medical informatics 6}, 2 (2018).

\bibitem{CensusDataset}
{\sc Kohavi, R.}
\newblock Scaling up the accuracy of naive-bayes classifiers: A decision-tree
  hybrid.
\newblock In {\em Proceedings of the Second International Conference on
  Knowledge Discovery and Data Mining\/} (1996), KDD'96, AAAI Press,
  pp.~202--207.

\bibitem{FederatedLearning}
{\sc Konečný, J., McMahan, H.~B., Yu, F.~X., Richtarik, P., Suresh, A.~T.,
  and Bacon, D.}
\newblock Federated learning: Strategies for improving communication
  efficiency.
\newblock In {\em NIPS Workshop on Private Multi-Party Machine Learning\/}
  (2016).

\bibitem{liu-onn-ccs-2017}
{\sc Liu, J., Juuti, M., Lu, Y., and Asokan, N.}
\newblock Oblivious neural network predictions via minionn transformations.
\newblock In {\em Proceedings of the 2017 ACM SIGSAC Conference on Computer and
  Communications Security\/} (New York, NY, USA, 2017), CCS '17, ACM,
  pp.~619--631.

\bibitem{FedTransLearning}
{\sc Liu, Y., Chen, T., and Yang, Q.}
\newblock Secure federated transfer learning.
\newblock {\em CoRR abs/1812.03337\/} (2018).

\bibitem{secure-agg-cacr2018}
{\sc Mandal, K., Gong, G., and Liu, C.}
\newblock Nike-based fast privacy-preserving high-dimensional data aggregation
  for mobile devices.
\newblock \textsc{CACR} \textsc{T}echnical \textsc{R}eport, \textsc{CACR}
  2018-10, University of Waterloo, Canada, 2018.

\bibitem{McMahan-model-averaging-deep-learning}
{\sc McMahan, H.~B., Moore, E., Ramage, D., Hampson, S., et~al.}
\newblock Communication-efficient learning of deep networks from decentralized
  data.
\newblock {\em arXiv preprint arXiv:1602.05629\/} (2016).

\bibitem{McMahan-model-averaging}
{\sc McMahan, H.~B., Moore, E., Ramage, D., and y~Arcas, B.~A.}
\newblock Federated learning of deep networks using model averaging.
\newblock {\em CoRR abs/1602.05629\/} (2016).

\bibitem{ABY3}
{\sc Mohassel, P., and Rindal, P.}
\newblock Aby3: A mixed protocol framework for machine learning.
\newblock In {\em Proceedings of the 2018 ACM SIGSAC Conference on Computer and
  Communications Security\/} (New York, NY, USA, 2018), CCS '18, ACM,
  pp.~35--52.

\bibitem{SecureML:SP}
{\sc Mohassel, P., and Zhang, Y.}
\newblock Secureml: A system for scalable privacy-preserving machine learning.
\newblock Cryptology ePrint Archive, Report 2017/396, 2017.
\newblock \url{http://eprint.iacr.org/2017/396}.

\bibitem{Nikolaenko:2013SP}
{\sc Nikolaenko, V., Weinsberg, U., Ioannidis, S., Joye, M., Boneh, D., and
  Taft, N.}
\newblock Privacy-preserving ridge regression on hundreds of millions of
  records.
\newblock In {\em Proceedings of the 2013 IEEE Symposium on Security and
  Privacy\/} (Washington, DC, USA, 2013), SP '13, IEEE Computer Society,
  pp.~334--348.

\bibitem{openssl}
{\sc OpenSSL.}
\newblock The openssl library.
\newblock \url{https://www.openssl.org/}.

\bibitem{Paillier}
{\sc Paillier, P.}
\newblock Public-key cryptosystems based on composite degree residuosity
  classes.
\newblock In {\em Proceedings of the 17th International Conference on Theory
  and Application of Cryptographic Techniques\/} (Berlin, Heidelberg, 1999),
  EUROCRYPT'99, Springer-Verlag, pp.~223--238.

\bibitem{sklearn}
{\sc Pedregosa, F., Varoquaux, G., Gramfort, A., Michel, V., Thirion, B.,
  Grisel, O., Blondel, M., Prettenhofer, P., Weiss, R., Dubourg, V.,
  Vanderplas, J., Passos, A., Cournapeau, D., Brucher, M., Perrot, M., and
  Duchesnay, E.}
\newblock Scikit-learn: Machine learning in python.
\newblock {\em J. Mach. Learn. Res. 12\/} (Nov. 2011), 2825--2830.

\bibitem{priv-reg-kdd2004}
{\sc Sanil, A.~P., Karr, A.~F., Lin, X., and Reiter, J.~P.}
\newblock Privacy preserving regression modelling via distributed computation.
\newblock In {\em Proceedings of the Tenth ACM SIGKDD International Conference
  on Knowledge Discovery and Data Mining\/} (New York, NY, USA, 2004), KDD '04,
  ACM, pp.~677--682.

\bibitem{Shamir:SS}
{\sc Shamir, A.}
\newblock How to share a secret.
\newblock {\em Commun. ACM 22}, 11 (Nov. 1979), 612--613.

\bibitem{PPDL}
{\sc Shokri, R., and Shmatikov, V.}
\newblock Privacy-preserving deep learning.
\newblock In {\em Proceedings of the 22Nd ACM SIGSAC Conference on Computer and
  Communications Security\/} (New York, NY, USA, 2015), CCS '15, ACM,
  pp.~1310--1321.

\bibitem{ModelStealingAttack}
{\sc Tram{\`e}r, F., Zhang, F., Juels, A., Reiter, M.~K., and Ristenpart, T.}
\newblock Stealing machine learning models via prediction apis.
\newblock In {\em 25th {USENIX} Security Symposium ({USENIX} Security 16)\/}
  (Austin, TX, Aug. 2016), {USENIX} Association, pp.~601--618.

\bibitem{FedLearning-HE-DP2018}
{\sc Truex, S., Baracaldo, N., Anwar, A., Steinke, T., Ludwig, H., and Zhang,
  R.}
\newblock A hybrid approach to privacy-preserving federated learning.
\newblock {\em CoRR abs/1812.03224\/} (2018).

\bibitem{Telemonitoring}
{\sc {Tsanas}, A., {Little}, M.~A., {McSharry}, P.~E., and {Ramig}, L.~O.}
\newblock Accurate telemonitoring of parkinson's disease progression by
  noninvasive speech tests.
\newblock {\em IEEE Transactions on Biomedical Engineering 57}, 4 (April 2010),
  884--893.

\bibitem{EnergyEfficiency}
{\sc Tsanas, A., and Xifara, A.}
\newblock Accurate quantitative estimation of energy performance of residential
  buildings using statistical machine learning tools.
\newblock {\em Energy and Buildings 49\/} (2012), 560 -- 567.

\bibitem{DuAtallah}
{\sc {Wenliang Du}, and {Atallah}, M.~J.}
\newblock Privacy-preserving cooperative statistical analysis.
\newblock In {\em Seventeenth Annual Computer Security Applications
  Conference\/} (Dec 2001), pp.~102--110.

\bibitem{GarbledCircuit}
{\sc Yao, A. C.-C.}
\newblock How to generate and exchange secrets.
\newblock In {\em Proceedings of the 27th Annual Symposium on Foundations of
  Computer Science\/} (Washington, DC, USA, 1986), SFCS '86, IEEE Computer
  Society, pp.~162--167.

\bibitem{CreditCardClient}
{\sc Yeh, I.-C., and hui Lien, C.}
\newblock The comparisons of data mining techniques for the predictive accuracy
  of probability of default of credit card clients.
\newblock {\em Expert Systems with Applications 36}, 2, Part 1 (2009), 2473 --
  2480.

\bibitem{LogisReg-Cloud}
{\sc Zhu, X.~D., Li, H., and Li, F.~H.}
\newblock Privacy-preserving logistic regression outsourcing in cloud
  computing.
\newblock {\em Int. J. Grid Util. Comput. 4}, 2/3 (Sept. 2013), 144--150.

\end{thebibliography}

%
%
%
%
%
%
\appendix
\vspace{-0.3cm}

\section{Cryptographic Background}\label{sec:CryptoBkgd}
In this section, we provide a background on the crypto-primitives and schemes that are needed for the aggregation protocol. We also present the aggregation privacy  game of Buescher et al. \cite{TwoIsNotEnough}.  
\subsection{PKI and DH Key Agreement} 
\regsys\ requires an existence of a public key infrastructure (PKI) due to the use of the secure aggregation protocol in which the users need to establish two pairwise keys, and each user and the server need to establish a pairwise key, which is done using the Diffie-Hellman (DH) protocol \cite{DHKeyAgreement}. For the users, one pairwise key is used to derive one-time pairwise keys  and the other one is used to realize an authenticated channel with other users to exchange information in the aggregation protocol. 
The users register their identities when they join first time and publish the identities and public keys  in the setup phase, as used in \cite{SecureAggrePPML,secure-agg-cacr2018}.

A key agreement protocol consists of a tuple of three algorithms $(\textsf{ParamGen}, \textsf{KeyGen}, \textsf{KeyAgree} )$. We use an Elliptic curve variant of the DH key agreement protocol. Given a security parameter $\kappa$, the parameter generation algorithm $(\Grp, q, G)\leftarrow \textsf{ParamGen}(1^{\kappa})$ samples a group of points on an elliptic curve over a field $\F_{q}$ of order $q$ and a generator $G$. 
The key generation algorithm $(x_A, x_A G)\leftarrow \textsf{KeyGen} (\Grp, q, P)$ generates a secret key $x_A \leftarrow^{\$} \Z_q$ and a public key $G_A = x_A G$.  The  key agreement function $\textsf{KeyAgree}$ computes a pairwise key $K_{AB} = H( G_{AB} ) $ where $G_{AB} \leftarrow  \textsf{KeyAgree} (x_A, G_B) = x_Ax_B G$ and $H$ is a hash function. The security of the key agreement protocol follows from the Decisional Diffie-Hellman (DDH) assumption.  

\vspace{-0.3cm}
\subsection{Authenticated Encryption}
An authenticated encryption scheme consists of three algorithms $(\mathcal{K}, \Enc, \Dec)$ where  $\mathcal{K}$ is a symmetric key generation algorithm, 
$\Enc$ is a symmetric-key encryption algorithm that accepts a key and a message as input and outputs a ciphertext and a tag, and $\Dec$ is a decryption algorithm that takes a key, a ciphertext and a tag, and outputs the original message or $\perp$. We use the \textsf{AES-GCM} that offers the IND-CPA and INT-CTXT securities to realize an authenticated channel.   
\vspace{-0.2cm}
\subsection{Pseudorandom Generator and One-way Function} 
A pseudorandom generator $\textsf{PRG}: \{0,1\}^l \rightarrow \{0,1\}^b$, is a deterministic function which accepts an $l$-bit binary string as input and outputs an $b$-bit binary string with $b>>l$. 
The input to a \textsf{PRG} is chosen uniformly at random. The security of a cryptographically strong $\textsf{PRG}$ is defined as its output bit stream is indistinguishable from a truly random one. 
We used \textsf{AES} in counter mode as a \textsf{PRG} in our experiment.
A one-way function $F: \{0,1\}^m \rightarrow \{0,1\}^{n}$ is a function that is easy to compute and hard to invert where $n \leq m$.  
Security of the one-way function $F$ is the amount of work needed to invert $F$ on an input $x$. 
We use the \textsf{SHA256} hash function to realize a one-way function.   
\vspace{-0.3cm}
\subsection{Threshold Secret Sharing}
Let $\F_q$ be a finite field where $q$ is a prime. A $t$-out-of-$m$ threshold secret sharing scheme ($(t,m)$-tss) over $\F_q$ is a tuple of two algorithms $(\sstm, \rectm)$ where the secret sharing algorithm $\textbf{Share}_m^t$ is a randomized algorithm that takes a secret input $s$ over $\F_q$ and outputs $m$ shares of $s$, i.e., $(s_1, s_2, \dots, s_m) \leftarrow \sstm(s)$ and the reconstruction algorithm $\rectm$ accepts any $t$ shares $(s_{i_1}, s_{i_2}, \cdots, s_{i_t})$ as input and outputs the original $s$. The correctness of the scheme is defined as 
for any $s$, $\displaystyle\textrm{Pr}_{(s_1, \cdots, s_m)\leftarrow \sstm(s)}[ \rectm( s_{i_1}, s_{i_2}, \cdots, s_{i_t} )=s] = 1$. The security of the scheme is defined as any set of less than $t$ users learns nothing about the secret $s$. 
We use Shamir's threshold secret sharing scheme \cite{Shamir:SS} and a $(2,3)$-tss in the secure aggregation protocol.
\vspace{-0.3cm}
\subsection{Aggregation Privacy Game}\label{sec:AggPrivGame}
Buescher et al. \cite{TwoIsNotEnough} formalized the privacy of an aggregation scheme  where the privacy of an aggregate is measured by the chances of an adversary in winning the game in Figure~\ref{fig:AggGame}. 
We use an aggregation protocol that satisfies the aggregation privacy game. 
\begin{figure}[h]
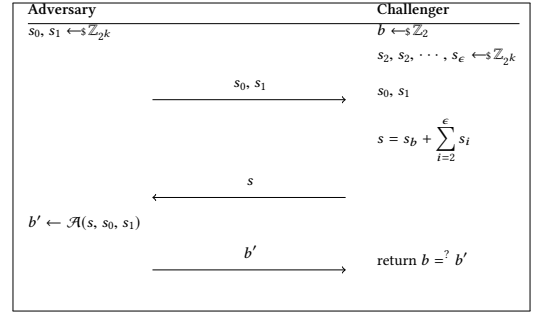

\centering
\vspace{-0.15cm}
\resizebox{7cm}{!}{
\fbox{
\pseudocode{%
{\bf Adversary} \< \< {\bf Challenger}\\ [0.1\baselineskip] [\hline]
s_0, s_1 \sample \Z_{2^k} \<\< b \sample \Z_2\\ 
 \<\< s_2, s_2, \cdots, s_{\epsilon} \sample \Z_{2^k} \\
 \< \sendmessageright*{ s_0, s_1 } \< s_0, s_1 \\
 \<\< s = s_b + \sum_{i = 2}^{\epsilon} s_i \\
 \< \sendmessageleft*{ s } \< \\
 b' \leftarrow \mathcal{A}(s, s_0, s_1) \<\< \\
 \< \sendmessageright*{ b' } \< \text{return } b =^{?} b' \\
}
}
}
\vspace{-0.35cm}
\caption{Privacy game for an aggregation scheme}
\label{fig:AggGame}
\vspace{-0.5cm}
\end{figure}
\begin{definition}
 We say an aggregate-sum is $\epsilon$-\emph{private} if the aggregate $s = \sum_{i = 1}^{\epsilon} s_i$ is secure under the privacy game defined in Figure~\ref{fig:AggGame}, meaning it leaks no information about any individual $s_i$. 
\end{definition}

\section{Security Analysis of \regsys}\label{sec:SecProof}
This section presents the security proofs of the shared local gradient computation protocols in Section~\ref{sec:PPGD}, and the training protocols in Section~\ref{sec:PPTraningProtocol} for linear and logistic regression models.  
\subsection{Security of Shared Local Gradient Computation Protocols}\label{sec:SLGSecProof}
Before providing the proofs, we first define the ideal functionality of the shared local gradient computation protocol in Figure~\ref{fig:IdealSLG}. 
We denote by $\func_{\textsc{LinSLG}}$ and $\func_{\textsc{LogSLG}}$ the ideal functionalities for the SLG for linear and logistic regression models, respectively.  
\begin{figure}[h]
\vspace{-0.3cm}
\begin{center}
\fbox{
\pseudocode[syntaxhighlight=auto,space=auto,width=7.5cm, mode=text]{
\textbf{Inputs} \\
\t\t{\bf Server} ($S$): Model $\thetabu$  \\
\t\t{\bf User} ($P_i$): Dataset $\DS_i$ \\
\textbf{Outputs} \\
\t\t{\bf User} ($P_i$): a random vector $\rbu_i$ \\
\t\t{\bf Server} ($S$): $\sbu_i$ such that $\sbu_i = \texttt{GD}(\thetabu, (\DS_i, |\DS_i|)) + \rbu_i$.
}
}
\vspace{-0.25cm}
\caption{Ideal functionality $\func_{\textsc{SLG}}$ for an SLG computation.} 
\label{fig:IdealSLG}
\end{center}
\vspace{-0.35cm}
\end{figure}

\par\noindent{\bf Proofs of Theorems~\ref{thm:LinSLG} and~\ref{thm:LogSLG}.} We prove the security of two SLG protocols in the simulation paradigm.
\begin{proof}[Proof of Theorem~\ref{thm:LinSLG}]
To prove $\pi_{\textsc{LinSLG}}$ is secure, we show that the adversary's ($\A$) view in the real-world execution of the protocol is indistinguishable from that of the ideal-world execution. 
We assume that a trusted party computes the ideal functionality $\func_{\textsc{LinSLG}}$ in the ideal-world execution of the protocol. We construct two different simulators, denoted by $\siml_S$ for the server $S$, and $\siml_{P_i}$ for the user $P_i$. 
We assume that the server receives the model $\thetabu$ and the user $P_i$ receives $\DS_i$ with size $d_i$ from the environment, and the simulator has access to $\func_{\textsc{LinSLG}}$.
\par\noindent{\bf Server $S$ is semi-honest.}  When $\A$ controls the server, the simulator $\siml_S$ works as follows: 
\begin{itemize}
\item $\siml_S$ receives $\thetabu$ from the environment and sends it to $\func_{\textsc{LinSLG}}$ and receives $\sbu' = \omegabu_i + \rbu' = (s_0', s_1', \cdots, s_n')$. 
\item $\siml_S$ runs $S$ on input $\thetabu$ and receives $\E(\thetabu) = (\E(\theta_0), \cdots, \E(\theta_n))$.
\item $\siml_S$ encrypts $\sbu$ using $S$'s public key and outputs whatever $S$ outputs. 
\end{itemize}
$S$'s view in the real-world execution of the protocol is $\sbu = (s_0, s_1, \dots, s_n)$ after decrypting $\E(\sbu)$ where $$\sbu = \omegabu_i + \rbu= \Big( \sum_{j = 1}^{d_i} e_j + r_0,  \cdots, \sum_{j = 1}^{d_i} e_jx_n^{(j)} + r_n \Big)$$
where $e_i = h(\thetabu, \xbu^{(j)})$. As $\rbu$ is chosen uniformly at random, any coordinate of $\sbu$, i.e., $s_l = \sum_{j = 1}^{d_i} e_jx_l^{(j)} + r_l, 0 \leq l \leq n $ is indistinguishable from random $s_l'$. 
Therefore, the server's view in the real-execution is indistinguishable from the view of the ideal-execution of the protocol. 
\par\noindent{\bf User $P_i$ is semi-honest.}  
When $\A$ controls the user $P_i$, the simulator $\siml_{P_i}$  works as follows.
\begin{itemize}
\item $\siml_{P_i}$ receives $\DS_i$ from the environment and sends it to $\func_{\textsc{LinSLG}}$ and receives $\rbu$. 
\item $\siml_{P_i}$ runs $P_i$ on the input $\DS_i$. 
\item $\siml_{P_i}$ constructs $\E(\textbf{0})$ and sends it to $P_i$ and outputs whatever $P_i$ outputs. 
\end{itemize}
$P_i$'s view in the real-world execution of the protocol is $\E(\thetabu) = ( \E(\theta_0), \cdots, \E(\theta_n) ) $ where each $\E(\theta_i)$ is computationally indistinguishable 
from $\E(0)$, resulting in $\E(\thetabu)$ is indistinguishable from $\E(\textbf{0})$ due to the semantic security of $\E$. 
Thus, the view of the adversary in the real-world execution of the protocol is indistinguishable from that of the ideal-world execution of the protocol.
\end{proof}
\vspace{-0.3cm}
%
\begin{proof}[Proof of Theorem~\ref{thm:LogSLG}]
We provide the constructions of two different simulators, denoted by $\siml_S$ for the server $S$ and $\siml_{P_i}$ for the user $P_i$. 
We assume that the simulator has access to the ideal functionality $\func_{\textsc{LogSLG}}$. 
We assume that the server receives the model $\thetabu$ and the user $P_i$ receives $\DS_i$ with size $d_i$ from the environment. 
\par\noindent{\bf Server $S$ is semi-honest.}  When $\A$ controls the server, the simulator $\siml_S$ works as follows: 
\begin{itemize}
\item $\siml_S$ receives $\thetabu$ from the environment and sends it to $\func_{\textsc{LinSLG}}$ and receives $\sbu' = \omegabu_i + \rbu' = (s_0', s_1', \cdots, s_n')$. 
\item $\siml_S$ runs $S$ on input $\thetabu$ and receives $\E(\thetabu) = (\E(\theta_0), \cdots, \E(\theta_n))$.
\item $\siml_S$ constructs $\big\{ z_j'\big\}_{j = 1}^{d_i}$ where $z_j' = \thetabu \cdot \xbu^{j} + c_j'$ and sends $\big\{\E(z_j')\big\}$ to $S$. Note that the simulator does not know actual values of $\big\{c_j'\big\}$ chosen by $P_i$. 
\item $\siml_S$ receives $\big\{z_j\big\}$ from $S$ and computes $\big\{z_j^2\big\}$ and $\big\{\sigma_3(z_j)\big\}$ as well as $\big\{\E(z_j^2), \E(\sigma_3(z_j))\big\}$ and sends it to $S$ 
\item $\siml_S$ encrypts $\sbu$ using $S$'s public key $pk$ and sends to $S$ and outputs whatever $S$ outputs. 
\end{itemize}
S's view in the real-world execution of the protocol is $\big\{z_j\big\}_{j = 1}^{d_i} \text{ and } \sbu = (s_0, s_1, \dots, s_n)$ after decrypting $\E(\sbu)$ where 
$$ z_j = \thetabu \cdot \xbu^{(j)} + c_j, \text{ and } \sbu = \omegabu_i + \rbu= \Big( \sum_{j = 1}^{d_i} e_j + r_0,  \cdots, \sum_{j = 1}^{d_i} e_jx_n^{(j)} + r_n \Big)$$
with $e_i = h(\thetabu, \xbu^{(j)})$ and $h$ is a logistic regression function. As $\rbu$ is chosen uniformly at random, any coordinate of $\sbu$, i.e., $s_l = \sum_{j = 1}^{d_i} e_jx_l^{(j)} + r_l, 0 \leq l \leq n $ is indistinguishable from $s_l'$. 
Therefore, the server's view in the real-world execution is indistinguishable from that of the ideal-world execution of the protocol. 
\par\noindent{\bf User $P_i$ is semi-honest.}  
When $\A$ controls the user $P_i$, the simulator $\siml_{P_i}$  works as follows.
\begin{itemize}
\item $\siml_{P_i}$ receives $\DS_i$ from the environment and sends it to $\func_{\textsc{LogSLG}}$ and receives $\rbu$. 
\item $\siml_{P_i}$ constructs $\E(\thetabu) = (\E(0), \cdots, \E(0))$ and sends it to $P_i$. 
\item $\siml_{P_i}$ runs $P_i$ on the inputs $\DS_i$ and $\E(\thetabu)$ and obtains $\big\{\E(z_j)\big\}$ where  $\siml_{P_i}$ generates random coins $\{c_j'\}$ for $P_i$. 
\item $\siml_{P_i}$ constructs $\big\{(\E(0), \E(0)) \big\}_{j = 1}^{d_i}$ and sends to $P_i$ and outputs whatever $P_i$ outputs. 
\end{itemize}
$P_i$'s view in the real-world execution of the protocol is given by $\E(\thetabu), \big\{ (\E(z_j^2), \E(\sigma_3(z_j)) \big\}_{j = 1}^{d_i} $,
which is computationally indistinguishable from that of the real-world execution of the protocol, due to the semantic security of $\E$. 
\end{proof}

%

\subsection{Security of Training Protocols}\label{sec:TrainSecProof}
In this section, we prove the security of the training protocol against semi-honest adversaries in the simulation paradigm by providing a construction of a simulator. 
We use an aggregation protocol $\pi_{\textsc{DeA}}$ that securely implements the ideal functionality in Figure~\ref{fig:IdealAgg}.
In \cite{SecureAggrePPML,secure-agg-cacr2018}, the security of the aggregation protocol is proved against semi-honest and active adversaries. 
We only need the security of $\pi_{\textsc{DeA}}$ against semi-honest adversaries. 
The security of the training protocol in the server-only threat model is a special case of Theorem~\ref{thm:SecProofTrain}. So we omit it. 
\begin{figure}[h]
\vspace{-0.3cm}
\begin{center}
\fbox{
\pseudocode[syntaxhighlight=auto,space=auto,width=7.5cm, mode=text]{
\textbf{Inputs} \\
\t\t{\bf Public param}: a set of users $\us$ with $|\us| = m$, a threshold value $t < m$ \\
\t\t{\bf Server} ($S$): Nothing ($\perp$)  \\
\t\t{\bf User} ($P_i$): Private input $\sbu_i, i \in \us$ \\
\textbf{Outputs} \\
\t\t{\bf User} ($P_i$): Nothing ($\perp$) \\
\t\t{\bf Server} ($S$): $\sbu = \sum_{u \in \vs}\sbu_u, \vs \subseteq \us$ if $|\vs| \geq t$ and $\perp$ otherwise. 
}
}
\vspace{-0.25cm}
\caption{Ideal functionality $\func_{\textsc{Agg}}$.}
\label{fig:IdealAgg}
\end{center}
\vspace{-0.35cm}
\end{figure}

\begin{lemma}\label{lem:GradNoLeak}
Assume that the aggregate-sum is $\epsilon$-private. For the regression models, the gradient $\omegabu$ on $\{\DS_i\}_{i \in \mathcal{B}}$ and $\mathcal{B} \subseteq [m]$ with $|\mathcal{B}|\geq 2t$ does not leak any information about an honest user's $\DS_i$ with $|\DS_i| \geq 1$ for a coalition of up to $(t-1)$ users and the server, where $t = \lceil\frac{m}{3}\rceil$. 
\end{lemma}
\begin{proof}
For the linear and logistic regression, the local gradient $\omegabu_i$ on $\DS_i$ with size $d_i$ can be written as 
$$\omegabu_i = \Big( \sum_{j = 1}^{d_i} e^{(j)}, \sum_{j = 1}^{d_i} e^{(j)}x_1^{(j)}, \cdots, \sum_{j = 1}^{d_i} e^{(j)}x_n^{(j)} \Big)$$
where $e^{(j)} = (h(\thetabu, \xbu^{(j)})-y^{(j)})$ and $h$ is a linear or logistic regression function. 
The gradient $\omegabu = (\omega_0, \cdots, \omega_n)$ on the dataset $\{\DS_i\}_{i \in \mathcal{B}}$ is given by 
$$\omegabu = \sum_{i \in \mathcal{B}}\omegabu_i = \Big( \sum_{i \in \mathcal{B}}\sum_{j = 1}^{d_i} e^{(j)}, \sum_{i \in \mathcal{B}}\sum_{j = 1}^{d_i} e^{(j)}x_1^{(j)}, \cdots, \sum_{i \in \mathcal{B}}\sum_{j = 1}^{d_i} e^{(j)}x_n^{(j)} \Big).$$
The first component of $\omegabu$ can be written as 
$\omega_0 = \sum_{i \in \mathcal{B}\backslash\Cp}\sum_{j = 1}^{d_i} e^{(j)} + \sum_{i \in \Cp}\sum_{j = 1}^{d_i} e^{(j)}$. 
For a set of corrupted users of size up to $t-1$, the number of terms in  $\sum_{i \in \mathcal{B}\backslash\Cp}\sum_{j = 1}^{d_i} e^{(j)}$ is $\sum_{i \in \mathcal{B}\backslash\Cp}d_i$, and similarly for other components. 
In worst case, when $|\DS_i| = 1$ and the number of the honest users in $\mathcal{B}\backslash\Cp$ is at least $t \geq \epsilon$, the gradient leaks no information about $\omegabu$ due to $\epsilon$-privacy of the aggregate. 
\end{proof}

\par\noindent{\bf Proofs of Theorems~\ref{thm:SecProofTrain} and~\ref{thm:SecProofUserOnlyTrain}.} Below we provide the proofs of the security theorems of the training protocols.  
\begin{proof}[Proof of Theorem~\ref{thm:SecProofTrain}]
We will prove the security of the protocols in $(\func_{\textsc{SLG}}, \func_{\textsc{Agg}})$-hybrid model using the standard hybrid argument. 
For $\pi_{\textsc{LinTrain}}$ and $\pi_{\textsc{LogTrain}}$, the ideal functionality $\func_{\textsc{SLG}}$ is replaced by $\func_{\textsc{LinSLG}}$ and $\func_{\textsc{LogSLG}}$, respectively.  

We provide a construction of a simulator through a sequence of hybrids which are constructed by subsequent modifications and argue that every two subsequent hybrids are computationally indistinguishable. 
The simulator $\siml$ runs the adversary $\A$ internally and provides the corrupted users inputs, and can emulates the honest parties' inputs as the actual inputs of the honest parties are unknown. 
$\siml$ has access to $\func_{\textsc{SLG}}$ and $\func_{\textsc{Agg}}$. We denote by $\Cp$ the set of corrupted parties and $|\Cp| \leq t$. 
\begin{itemize}
 \item[{\bf Hyb 0}:] This hybrid is a random variable corresponding to the joint view of $\A$ in the real-world execution of the protocol. 
 \item[{\bf Hyb 1}:] This hybrid is identically same as the previous one, except the key agreement phase. For the honest users in $\us_0\backslash \mathcal{C}$, instead of using the DH key agreement algorithm (\textsf{KeyAgree}), $\siml$ uses a pair of uniformly random keys for encryption/decryption and one-time key generation. 
 The Decisional Diffie-Hellman assumption ensures that this hybrid is indistinguishable from the previous one. 
 \item[{\bf Hyb 2}:] Note that $\Xbu_{sum} = \Xbu_{sum}^{\Cp} + \Xbu_{sum}^{\us_1\backslash \Cp}$ where 
 $\Xbu_{sum}^{ \Cp}$ is the sum of the corrupted user inputs and $\Xbu_{sum}^{\us_1\backslash \Cp}$ is the sum of the honest users inputs. When $\Xbu_{sum} = \perp$, $\siml$ aborts. 
 In this hybrid, $\siml$ samples $\{Z_u\}_{u \in \us_1\backslash \Cp}$ such that $\sum_{u \in \us_1\backslash \Cp} Z_u = \Xbu_{sum}$. 
 Instead of sending $\{X_u\}_{u \in \us_1\backslash \Cp}$ as inputs for the honest users, $\siml$ sends $\{Z_u\}_{u \in \us_1\backslash \Cp}$ as inputs to $\func_{\textsc{Agg}}$. 
 As $|\us_1\backslash \Cp| \geq t$, the $\epsilon$-privacy of the aggregation-sum ensures that the distributions $\{Z_u\}_{u \in \us_1\backslash \Cp}$ and $\{X_u\}_{u \in \us_1\backslash \Cp}$ are identical, where the number of inputs in the sum of the honest users' inputs is $\ell t \geq \epsilon$. 
 Thus, this hybrid is indistinguishable from {\bf Hyb 1}.
 
 \item[{\bf Hyb 3}:] In this hybrid, we change the encryption of $\thetabu$ by $\E(\textbf{0})$. $\A$'s view in the real-world execution of the protocol contains $\E(\thetabu)$. The semantic security of the encryption scheme ensures that this hybrid is indistinguishable from the previous one. 
 \item[{\bf Hyb 4}:] In this hybrid, for each honest user, the input $\DS_u$ is sampled as $\DS_u'$ s.t. 
 $Z_u = \big( \sum_{j = 1}^{|\DS_u'|} {x_{1}^{(j)'}}, \cdots,  \sum_{j = 1}^{|\DS_u'|} {x_{n}^{(j)'}}^2 \big)$, $u \in \us_1\backslash \Cp$. Note that $\siml$ uses the knowledge of $\{Z_u\}$ in {\bf Hyb 2}. 
 $\siml$ constructs $\vs$ by randomly choosing $M$ users with $M \geq 2t$. $\siml$ aborts if $|\vs| < M$. 
 For each honest user in $\vs\backslash \Cp$, $\siml$ sends $\DS_u'$ and $\thetabu$ to $\func_{\textsc{SLG}}$ and receives a random share $\sbu_u'$. 
 $\siml$ runs polynomial-many $\func_{\textsc{SLG}}$ and obtains $\{ \sbu_u' \}_{u \in \vs\backslash \Cp}$. $\A$'s view in the real-world execution contains 
 $\{ \sbu_u \}_{u \in \vs\backslash \Cp}$ with $\sbu_u = \omegabu_u + \rbu_u$ where $\{\rbu_u\}$ is randomly generated. Thus, the distribution $\{ \sbu_u' \}_{u \in \vs\backslash \Cp}$ is identically distributed to $\{ \sbu_u \}_{u \in \vs\backslash \Cp}$. 
Thus, this hybrid is indistinguishable from {\bf Hyb 3}. 
 \item[{\bf Hyb 5}:] In this hybrid, as  $\rbu_{sum} = \rbu_{sum}^{\Cp} + \rbu_{sum}^{\vs''\backslash \Cp}$, $\siml$ emulates the honest users' inputs $\big\{\rbu_u'\big\}_{u \in \vs''\backslash \Cp}$ such that $\sum_{u \in \vs''\backslash \Cp} \rbu_u' = \sum_{u \in \vs''\backslash \Cp} \rbu_u = \rbu_{sum}^{\vs''\backslash \Cp}$, 
 where $\big\{\rbu_u\big\}_{u \in \vs''\backslash \Cp}$ is randomly distributed, and 
 for the honest users in $\vs'\backslash(\vs''\backslash \Cp)$, it randomly samples arbitrary values and sends all to $\func_{\textsc{Agg}}$. 
 Note that $\omegabu = \omegabu^{\Cp} + \omegabu^{\vs''\backslash \Cp} = \sum_{u \in \Cp} (\sbu_u - \rbu_u) + \sum_{u \in \vs''\backslash \Cp} (\sbu_u - \rbu_u)$, and the number of honest users in $\vs''\backslash \Cp$ is $|\vs''\backslash \Cp| = M-\delta \geq \rho$. 
 Therefore, $\omegabu^{\vs''\backslash \Cp}$ is the sum of at least $\rho\cdot \ell = \epsilon$ inputs. 
 In the real-world execution, any leakage about an individual $\sbu_u$ can happen from the execution of the aggregation protocol with a negligible probability.
 The security of the aggregation protocol and the $\epsilon$-privacy of the aggregate-sum ensures that the distributions $\big\{\rbu_u'\big\}_{u \in \vs''\backslash \Cp}$ and $\big\{\rbu_u\big\}_{u \in \vs''\backslash \Cp}$ are identical and so do $\big\{\omegabu_u\big\}_{u \in \vs''\backslash \Cp}$ and $\big\{\omegabu_u'\big\}_{u \in \vs''\backslash \Cp}$ with $\sum_{u \in \vs''\backslash \Cp} \omegabu_u' = \omegabu^{\vs''\backslash \Cp}$ emulated by $\siml$. 
 Therefore, this hybrid is indistinguishable from the previous one. 
 \item[{\bf Hyb 6}:] We repeat {\bf Hyb 3} to {\bf Hyb 5} sequentially $(R-1)$ times (polynomial many times), and  it is easy to observe that each subsequent modification of the hybrids is indistinguishable, by applying the above arguments.
\end{itemize}
This concludes the construction of the simulator. Thus, the output of the simulator is computationally indistinguishable from the output of the real-world execution of the protocol. Hence the proof follows. 
\end{proof}

\begin{proof}[Proof of Theorem~\ref{thm:SecProofUserOnlyTrain}]
This can be proved by constructing a simulator in a similar way, as of Theorem~\ref{thm:SecProofTrain}. 
We only emphasize the main behavior changes of the simulator and omit other details. Since the adversary corrupts only the set of users in $\Cp$, the joint view of the adversary, after the data scaling phase, contains $\mubu$ and $\sigmabu$, which involves the honest users inputs. 
$\siml$ emulates the honest users' (in $\us_1\backslash \Cp$) inputs $\DS_u'$ such that $\sum_{u \in \Cp} W_u + \sum_{u \in \us_1\backslash \Cp}Y_{u} = d_1 \cdot \mubu$ where $  Y_u = \big( \sum_{j = 1}^{|\DS_u'|} x_{1}^{(j)'}, \cdots,  \sum_{j = 1}^{|\DS_u'|} x_{n}^{(j)'}\big)$ and $  W_u = \big( \sum_{j = 1}^{|\DS_u|} x_{1}^{(j)}, \cdots,  \sum_{j = 1}^{|\DS_u|} x_{n}^{(j)}\big)$ 
and uses the true inputs $\DS_u$ for the corrupted users $u \in \Cp$, and these inputs are used in the rest of the simulation. Note that in the training phase no output is received by the users in $\Cp$. The simulator can use a dummy vector for the sum of the honest users' local gradient $ \omegabu^{\vs''\backslash\Cp}$.
\end{proof}


%
%
%


\end{document}